%% file: main.tex
\documentclass{llncs}
\usepackage{makeidx}
\usepackage{subfigure}
\usepackage{amsmath,amssymb,amsfonts}
\usepackage{dsfont}
\usepackage{latexsym}
\usepackage{xspace}
\usepackage{graphicx}
\usepackage{fancybox}
\usepackage{hyperref}
\usepackage{paralist}
\usepackage{color}
\usepackage{wrapfig}
\usepackage{tikz}
\usepackage{setspace}
\usepackage{algorithm}
\usepackage[noend]{algpseudocode}
\usepackage[framemethod=tikz]{mdframed}

\title{On the Computational Complexity of Minimal Cumulative Cost Graph Pebbling}
\author{
Jeremiah Blocki\inst{1}
\and
Samson Zhou\inst{2}
}
\institute{Department of Computer Science, Purdue University, West Lafayette, IN. 
\newline Email: {\tt jblocki@purdue.edu}.
\and
Department of Computer Science, Purdue University, West Lafayette, IN. 
\newline Email: {\tt samsonzhou@gmail.com}. 
}

\input{macros}

\begin{document}
\maketitle
\input{abstract}
\input{intro}

\input{background}
\input{related}
\input{reduction}
\input{reducible}

\input{approximation}
\input{future}
\section*{Acknowledgements}
We would like to thank Ioana Bercea and anonymous reviewers for helpful comments that improved the presentation of the paper.
The work was supported by the National Science Foundation under NSF Awards $\#1649515$ and $\#1704587$. The opinions expressed in this paper are those of the authors and do not necessarily reflect those of the National Science Foundation.
\def\shortbib{0}
\vspace{-0.1in}
\bibliographystyle{alpha}
\bibliography{bounded-parallel-mhf,password,abbrev3,crypto}
\input{appendix}

\end{document}

%% file: macros.tex
\newcommand{\ignore}[1]{}

\newcommand{\COMMENTED}[1]{{}}

\def \cubicVC {\mdef{\mathtt{CubicVC}}}
\def \VC {\mdef{\mathtt{VC}}}
\def \ThreePartition {\mdef{\mathtt{3\-PARTITION}}}
\def \BTwoLC {\mdef{\mathtt{B2LC}}}
\def \minCC {\mdef{\mathtt{minCC}}}
\def \minST {\mdef{\mathtt{minST}}}
\def \minSST {\mdef{\mathtt{minSST}}}
\def \REDUCIBLE {\mdef{\mathtt{REDUCIBLE}}}
\def \MINREDUCIBLE {\mdef{\mathtt{minREDUCIBLE}}}
\newcommand{\mdef}[1]{{\ensuremath{#1}}\xspace}  

\newcommand{\superscript}[1]{\ensuremath{^{\mbox{\tiny{\textit{#1}}}}}\xspace}
\def \th {\superscript{th}}     
\def \etal{{\it et~al.}}

\newcommand{\ceil}[1]{\mdef{\left\lceil#1\right\rceil}}               

\newcommand{\node}[2]{\mdef{\langle#1, #2\rangle}}   

\def \indeg    {\mdef{\mathsf{indeg}}}
\def \sinks    {\mdef{\mathsf{sinks}}}

\def \depth    {\mdef{\mathsf{depth}}}

\def \parents  {\mdef{\mathsf{parents}}}

\renewcommand{\-}{\mbox{-}}     

\def \cc       {\mdef{\mathsf{cc}}}       
\def \ST     {\mdef{\mathsf{st}}}

\def \scrypt   {\mdef{\tt scrypt}}

\def \NP  {\mdef{\tt NP}}
\def \NPhard {\mdef{\mathtt{NP}\-\mathtt{Hard}}}
\def \NPcomplete {\mdef{\mathtt{NP}\-\mathtt{Complete}}}

\def \d        {\mdef{\delta}}                       

\newcommand{\namedref}[2]{\hyperref[#2]{#1~\ref*{#2}}}
\newcommand{\thmlab}[1]{\label{thm:#1}}
\newcommand{\thmref}[1]{\namedref{Theorem}{thm:#1}}
\newcommand{\lemmlab}[1]{\label{lemm:#1}}
\newcommand{\lemmref}[1]{\namedref{Lemma}{lemm:#1}}
\newcommand{\claimlab}[1]{\label{claim:#1}}
\newcommand{\claimref}[1]{\namedref{Claim}{claim:#1}}

\newcommand{\seclab}[1]{\label{sec:#1}}
\newcommand{\secref}[1]{\namedref{Section}{sec:#1}}
\newcommand{\applab}[1]{\label{app:#1}}
\newcommand{\appref}[1]{\namedref{Appendix}{app:#1}}

\newcommand{\figref}[1]{\namedref{Figure}{fig:#1}}

\newtheorem{Claim}[theorem]{Claim}
\newtheorem{fact}[theorem]{Fact}

\newenvironment{proofof}[1]{\noindent {\em Proof of #1.  }}{$\hfill\square$}

\newenvironment{remindertheorem}[1]{\medskip \noindent {\bf Reminder of  #1.  }\em}{}

\newenvironment{reminderlemma}[1]{\medskip \noindent {\bf Reminder of  #1.  }\em}{}

\newcommand{\peb}{\Pi} 
\newcommand{\Peb}{{\cal P}} 

\newcommand{\ppeb}{\peb^{\parallel}} 
\newcommand{\pPeb}{\Peb^{\parallel}}

\newcommand{\Ccc}{{cc}}
\newcommand{\Cspacetime}{{st}}

\def \pcc {\ppeb_\Ccc} 
\def \sst {\peb_{\Cspacetime}} 

%% file: abstract.tex
\begin{abstract}
We consider the computational complexity of finding a legal black pebbling of a DAG $G=(V,E)$ with minimum cumulative cost. A black pebbling is a sequence $P_0,\ldots, P_t \subseteq V$ of sets of nodes which must satisfy the following properties: $P_0 = \emptyset$ (we start off with no pebbles on $G$), $\sinks(G) \subseteq \bigcup_{j \leq t} P_j$ (every sink node was pebbled at some point) and $\parents\big(P_{i+1}\backslash P_i\big) \subseteq P_i$ (we can only place a new pebble on a node $v$ if all of $v$'s parents had a pebble during the last round). The cumulative cost of a pebbling $P_0,P_1,\ldots, P_t \subseteq V$ is $\cc(P) = \left| P_1\right| + \ldots + \left| P_t\right|$. The cumulative pebbling cost is an especially important security metric for data-independent memory hard functions, an important primitive for password hashing. Thus, an efficient (approximation) algorithm would be an invaluable tool for the cryptanalysis of password hash functions as it would provide an automated tool to establish tight bounds on the amortized space-time cost of computing the function.  We show that such a tool is unlikely to exist in the most general case. In particular, we prove the following results.
\begin{itemize}
\item It is $\NPhard$ to find a pebbling minimizing cumulative cost. 
\item The natural linear program relaxation for the problem has integrality gap $\tilde{O}(n)$, where $n$ is the number of nodes in $G$. We conjecture that the problem is hard to approximate.  
\item We show that a related problem, find the minimum size subset $S\subseteq V$ such that $\depth(G-S) \leq d$, is also $\NPhard$. In fact, under the Unique Games Conjecture there is no $(2-\epsilon)$-approximation algorithm.
\end{itemize}
\end{abstract}

%% file: intro.tex
\section{Introduction}\seclab{intro}
Given a directed acyclic graph (DAG) $G=(V,E)$ the goal of the (parallel) black pebbling game is to start with pebbles on some source nodes of $G$ and ultimately place pebbles on all sink nodes (not necessarily simultaneously). The game is played in rounds and we use $P_i \subseteq V$ to denote the set of currently pebbled nodes on round $i$. Initially all nodes are unpebbled, $P_0 = \emptyset$, and in each round $i \ge 1$ we may only include $v \in P_i$ if all of $v$'s parents were pebbled in the previous configuration ($\parents(v) \subseteq P_{i-1}$) or if $v$ was already pebbled in the last round ($v \in P_{i-1}$). In the sequential pebbling game we can place at most one new pebble on the graph in any round (i.e., $\left|P_i \backslash P_{i-1} \right| \leq 1)$, but in the parallel pebbling game no such restriction applies. 

Let $\pPeb_G$ (resp. $\Peb_G$) denote the set of all valid parallel (resp. sequential) pebblings of $G$. 
We define the cumulative cost (respectively space-time cost) of a pebbling $P=(P_1,\ldots,P_t) \in \pPeb_G$ to be $\cc(P)=|P_1|+\ldots+|P_t|$ (resp. $\ST(P) = t \times \max_{1 \le i \le t} \left| P_i\right|$), that is, the sum of the number of pebbles on the graph during every round.
The \emph{parallel} cumulative pebbling cost of $G$, denoted $\pcc(G)$ (resp. $\sst(G)= \min_{P \in \Peb_G} \ST(G)$), is the cumulative cost of the best legal pebbling of $G$. Formally,   \[ \pcc(G) = \min_{P \in \pPeb_G} \cc(P) \ , \mbox{~and} \qquad \qquad
\sst(G)=\min_{P\in\Peb_G}\ST(P)\ . \] 
 In this paper, we consider the computational complexity of  $\pcc(G)$, showing that the value is $\NPhard$ to compute. We also demonstrate that the natural linear programming relaxation for approximating $\pcc(G)$ has a large integrality gap and therefore any approximation algorithm likely requires more powerful techniques.
 
\subsection{Motivation}
The pebbling cost of a DAG $G$ is closely related to the cryptanalysis of data-independent memory hard functions (iMHF)~\cite{AS15}, a particularly useful primitive for password hashing~\cite{PHC,Argon2}. In particular, an efficient algorithm for (approximately) computing $\pcc(G)$ would enable us to automate the cryptanalysis of candidate iMHFs. The question is particularly timely as the Internet Research Task Force considers standardizing Argon2i~\cite{Argon2}, the winner of the password hashing competition~\cite{PHC}, despite recent attacks~\cite{CBS16,AlwenB16a,EC:AlwBloPie17} on the construction. Despite recent progress~\cite{AlwenB16b,EC:AlwBloPie17,TCC:BloZho17} the precise security of Argon2i and alternative constructions is poorly understood.
  
An iMHF is defined by a DAG $G$ (modeling data-dependencies) on $n$ nodes $V=\{1,\ldots,n\}$ and a compression function $H$ (usually modeled as a random oracle in theoretical analysis)\footnote{Because the data-dependencies in an iMHF are specified by a static graph, the induced memory access pattern does not depend on the secret input (e.g., password). This makes iMHFs resistant to side-channel attacks. Data-dependent memory hard functions (MHFs) like \scrypt~\cite{Per09} are potentially easier to construct, but they are potentially vulnerable to cache-timing attacks.  } . The label $\ell_1$ of the first node in the graph $G$ is simply the hash $H(x)$ of the input $x$. A vertex $ i > 1$ with parents $i_1<i_2<\cdots <i_\delta$ has label $\ell_i(x)=H(i,\ell_{i_1}(x),\ldots,\ell_{i_\delta}(x))$. The output value is the label $\ell_n$ of the last node in $G$. It is easy to see that any legal pebbling of $G$ corresponds to an algorithm computing the corresponding iMHF. Placing a new pebble on node $i$ corresponds to computing the label $\ell_i$ and keeping (resp. discarding) a pebble on node $i$ corresponds to storing the label in memory (resp. freeing memory). Alwen and Serbinenko~\cite{AS15} proved that in the parallel random oracle model (pROM) of computation, {\em any} algorithm evaluating such an iMHF could be reduced to a pebbling strategy with (approximately) the same cumulative memory cost.
 
It should be noted that {\em any} graph $G$ on $n$ nodes has a sequential pebbling strategy $P \in \Peb_G$ that finishes in $n$ rounds and has cost $\cc(P) \leq \ST(P) \leq n^2$.  Ideally, a good iMHF construction provides the guarantee that the amortized cost of computing the iMHF remains high (i.e., $\tilde{\Omega}\left( n^2 \right)$) even if the adversary evaluates many instances (e.g.,different  password guesses) of the iMHF. Unfortunately, neither large $\sst(G)$ nor large $\min_{P \in \pPeb_G} \ST(P)$, are sufficient to guarantee that $\pcc(G)$ is large~\cite{AS15}. More recently Alwen and Blocki~\cite{AlwenB16a} showed that Argon2i~\cite{Argon2}, the winner of the recently completed password hashing competition~\cite{PHC}, has much lower than desired amortized space-time complexity. In particular, $\pcc(G) \leq \tilde{O}\left(n^{1.75} \right)$. 

In the context of iMHFs, it is important to study $\pcc(G)$, the cumulative pebbling cost of a graph $G$, in addition to $\sst(G)$. Traditionally, pebbling strategies have been analyzed using space-time complexity or simply space complexity. While sequential space-time complexity may be a good model for the cost of computing a single-instance of an iMHF on a standard single-core machine (i.e., the costs incurred by the honest party during password authentication), it does not model the amortized costs of a (parallel) offline adversary who obtains a password hash value and would like to evaluate the hash function on {\em many} different inputs (e.g., password guesses) to crack the user's password~\cite{AS15,AlwenB16a}.  Unlike $\sst(G)$, $\pcc(G)$ models the {\em amortized} cost of evaluating a data-independent memory hard function on many instances~\cite{AS15,AlwenB16a}.
 
An efficient algorithm to (approximately) compute $\pcc(G)$ would be an incredible asset when developing and evaluating iMHFs. 
For example, the Argon2i designers argued that the Alwen-Blocki attack~\cite{AlwenB16a} was not particularly effective for practical values of $n$ (e.g., $n \leq 2^{20}$) because the constant overhead was too high \cite{Argon2}. 
However, they could not rule out the possibility that more efficient attacks might exist\footnote{Indeed, Alwen and Blocki~\cite{AlwenB16b} subsequently introduced heuristics to improve their attack and demonstrated that their attacks were effective even for smaller (practical) values of $n$ by simulating their attack against real Argon2i instances.}. 
As it stands, there is a huge gap between the best known upper/lower bounds on $\pcc(G)$ for Argon2i and for the new DRSample graph~\cite{CCS:AlwBloHar17}, since in all practical cases the ratio between the upper bound and the lower bound is at least $10^5$. 
An efficient algorithm to (approximately) compute $\pcc(G)$ would allow us to immediately resolve such debates by automatically  generating upper/lower bounds on the cost of computing the iMHF for each running time parameters ($n$) that one might select in practice. Alwen \etal~\cite{EC:AlwBloPie17} showed how to construct graphs $G$ with $\pcc(G) = \Omega\left(\frac{n^2}{\log n}  \right)$. 
This construction is essentially optimal in theory as results of Alwen and Blocki~\cite{AlwenB16a} imply that any constant indegree graph has $\pcc(G) = O\left(\frac{n^2 \log \log n}{\log n} \right)$. 
However, the exact constants one could obtain through a theoretical analysis are most-likely small. A proof that $\pcc(G) \geq \frac{10^{-6}\times n^2}{\log n}$ would be an underwhelming security guarantee in practice, where we may have $n \approx 10^6$. An efficient algorithm to compute $\pcc(G)$ would allow us to immediately determine whether these new constructions provide meaningful security guarantees in practice.

\subsection{Results} \seclab{results}
We provide a number of computational complexity results. 
Our primary contribution is a proof that the decision problem ``is $\pcc(G) \leq k$ for a positive integer $k \leq \frac{n(n+1)}{2}$'' is $\NPcomplete$.\footnote{\label{note1}Note that for any $G$ with $n$ nodes we have $\pcc(G) \leq  1+2+\ldots + n = \frac{n(n+1)}{2}$ since we can always pebble $G$ in topological order in $n$ steps if we never remove pebbles.  } In fact, our result holds even if the DAG $G$ has constant $\indeg$.\footnote{For practical reasons most iMHF candidates are based on a DAG $G$ with constant indegree.}  We also provide evidence that $\pcc(G)$ is hard to approximate. Thus, it is unlikely that it will be possible to automate the cryptanalysis process for iMHF candidates. In particular, we define a natural integer program to compute $\pcc(G)$ and consider its linear programming relaxation. We then show that the integrality gap is at least $\Omega\left(\frac{n}{\log n} \right)$ leading us to conjecture that it is hard to approximate $\pcc(G)$ within constant factors. We also give an example of a DAG $G$ on $n$ nodes with the property that {\em any} optimal pebbling (minimizing $\pcc$)  requires more than $n$ pebbling rounds. 
 
The computational complexity of several graph pebbling problems has been explored previously in various settings~\cite{gilbert1980pebbling,hertel2010pspace}. 
However, minimizing cumulative cost of a pebbling is a very different objective than minimizing the space-time cost or space. 
For example, consider a pebbling where the maximum number of pebbles used is significantly greater than the average number of pebbles used. 
Thus, we need fundamentally new ideas to construct appropriate gadgets for our reduction.\footnote{See additional discussion in \secref{relatedPebblingComplexity}.}
We first introduce a natural problem that arises from solving systems of linear equations, which we call \texttt{Bounded 2-Linear Covering} (\BTwoLC) and show that it is $\NPcomplete$. 
We then show that we can encode a $\BTwoLC$ instance as a graph pebbling problem thus proving that the decision version of cummulative graph pebbling is $\NPhard$.

In \secref{Reducible} we also investigate the computational complexity of determining how ``depth-reducible'' a DAG $G$ is showing that the problem is $\NPcomplete$ even if $G$ has constant indegree. A DAG $G$ is $(e,d)$-reducible if there exists a subset $S\subseteq V$ of size $|S| \leq e$ such that $\depth(G-S) < d$. That is, after removing nodes in the set $S$ from $G$, any remaining directed path has length less than $d$. If $G$ is not $(e,d)$-reducible, we say that it is $(e,d)$-depth robust. It is known that a graph has high cumulative cost (e.g., $\tilde{\Omega}\left(n^2 \right)$) if and only if the graph is highly depth robust (e.g., $e,d = \tilde{\Omega}\left( n \right)$)~\cite{AlwenB16a,EC:AlwBloPie17}. Our reduction from Vertex Cover preserves approximation hardness.\footnote{Note that when $d=0$ testing whether a graph $G$ is $(e,d)$ reducible is {\em equivalent} to asking whether $G$ has a vertex cover of size $e$. Our reduction establishes hardness for $d \gg 1$.}  Thus, assuming that $\mathtt{P} \ne \NP$ it is hard to $1.3$-approximate $e$, the minimum size of a set $S \subseteq V$ such that $\depth(G-S) < d$ ~\cite{dinur2005hardness}. Under the Unique Games Conjecture~\cite{khot2002power}, it is hard to $(2-\epsilon)$-approximate $e$ for any fixed $\epsilon>0$~\cite{khot2008vertex}. In fact, we show that the linear programming relaxation to the natural integer program to compute $e$ has an integrality gap of $\Omega(n/\log n)$ leading us to conjecture that it is hard to approximate $e$.


%% file: background.tex
\section{Preliminaries}\seclab{prelim}
Given a directed acyclic graph (DAG) $G=(V,E)$ and a node $v \in V$ we use $\parents(v) = \{u~:~(u,v) \in E \} $ to denote the set of nodes $u$ with directed edges into node $v$ and we use $\indeg(v) = \left|\parents(v)\right|$ to denote the number of directed edges into node $v$. We use $\indeg(G) = \max_{v \in V} \indeg(v)$ to denote the maximum indegree of any node in $G$. For convenience, we use $\indeg$ instead of $\indeg(G)$ when $G$ is clear from context. We say that a node $v \in V$ with $\indeg(v) = 0$ is a source node and a node with no outgoing edges is a sink node. We use $\mathtt{sinks}(G)$ (resp. $\mathtt{sources}(G)$) to denote the set of all sink nodes (resp. source nodes) in $G$. We will use $n=|V|$ to denote the number of nodes in a graph, and for convenience we will assume that the nodes $V=\{1,2,3,\ldots, n\}$ are given in topological order (i.e., $1 \leq j < i \leq n$ implies that $(i,j) \notin E$).  We use $\depth(G)$ to denote the length of the longest directed path in $G$. Given a positive integer $k \geq 1$ we will use $[k]= \{1,2,\ldots, k\}$ to denote the set of all integers $1$ to $k$ (inclusive).

\begin{definition}
Given a DAG $G = (V,E)$ on $n$ nodes a legal pebbling of $G$ is a sequence of sets $P=\big(P_0,\ldots,P_t\big)$ such that:
\begin{enumerate}
\item
$P_0 = \emptyset$
\item
$\forall i>0$, $v \in P_{i}\backslash P_{i-1}$ we have $\parents(v) \subseteq P_{i-1}$
\item
$\forall v \in \mathtt{sinks}(G)$ $\exists 0 < j \leq t$ such that $v \in P_j$
\end{enumerate}
The cumulative cost of the pebbling $P$ is $\cc(P) = \sum_{i=1}^t \left| P_i \right|$, and the space-time cost is $\ST(P) = t \times \max_{0 < j \leq t} \left| P_i \right|$.
\end{definition}

The first condition states that we start with no pebbles on the graph. The second condition states that we can only add a new pebble on node $v$ during round $i$ if we already had pebbles on all of $v$'s parents during round $i-1$. Finally, the last condition states that every sink node must have been pebbled during {\em some} round. 

We use $\pPeb_G$ to denote the set of all legal pebblings, and we use $\Peb_G \subset \pPeb_G$ to denote the set of all sequential pebblings with the additional requirement that $\left| P_i\backslash P_{i-1}\right| \leq 1$ (i.e., we place at most one new pebble on the graph during ever round $i$). We use $\pcc(G) = \min_{P \in \pPeb_G} \cc(P)$ to denote the cumulative cost of the best legal pebbling. 
 
 \begin{definition}
 We say that a directed acyclic graph (DAG) $G=(V,E)$ is $(e,d)$-depth robust if $\forall S \subseteq V$ of size $|S| \leq e$ we have $\depth(G-S) \geq d$. If $G$ contains a set $S \subseteq V$ of size $|S| \leq e$ such that $\depth(G-S) \leq d$ then we say that $G$ is $(e,d)$-reducible. 
 \end{definition}
 
 \subsection*{Decision Problems}

 The decision problem $\minCC$ is defined as follows: \\
 {\noindent \bf Input:} a DAG $G$ on $n$ nodes and an integer $k < n(n+1)/2$. \footnote{See footnote \ref{note1}.}\newline
 {\noindent \bf Output:} {\em Yes}, if $\pcc(G) \leq k$; otherwise {\em No}.
\newline\noindent

 Given a constant $\d \ge 1$ we use $\minCC_\d$ to denote the above decision problem with the additional constraint that $\indeg(G) \leq \d$. It is clear that $\minCC \in \NP$ and $\minCC_\d\in \NP$ since it is easy to verify that a candidate pebbling $P$ is legal and that $\cc(P) \leq k$.  
One of our primary results is to show that the decision problems $\minCC$ and $\minCC_2$ are $\NPcomplete$. In fact, these results hold even if we require that the DAG $G$ has a single sink node. 
 
 The decision problem $\REDUCIBLE_d$ is defined as follows: \\
  {\noindent \bf Input:} a DAG $G$ on $n$ nodes and positive integers $e,d \leq n$. \newline
 {\noindent \bf Output:} {\em Yes}, if $G$ is $(e,d)$-reducible; otherwise {\em No}. 

We show that the decision problem $\REDUCIBLE_d$ is $\NPcomplete$ for {\em all} $d > 0$ by a reduction from Cubic Vertex Cover, defined below. Note that when $d=0$ $\REDUCIBLE_d$ {\em is} Vertex Cover. We use $\REDUCIBLE_{d,\d}$ to denote the decision problem with the additional constraint that $\indeg(G) \leq \d$.

 The decision problem $\VC$ (resp. $\cubicVC$) is defined as follows: \\
  {\noindent \bf Input:} a graph $G$ on $n$ vertices ($\cubicVC$: each with degree $3$) and a positive integer $k \leq\frac{n}{2}$. \newline
 {\noindent \bf Output:} {\em Yes}, if $G$ has a vertex cover of size at most $k$; otherwise {\em No}.
 
 To show that $\minCC$ is $\NPcomplete$ we introduce a new decision problem $\BTwoLC$. We will show that the decision problem $\BTwoLC$ is $\NPcomplete$ and we will give a reduction from $\BTwoLC$ to $\minCC$.

 The decision problem  \textsf{Bounded 2-Linear Covering} (\BTwoLC) is defined as follows: \\
{\noindent \bf Input:} an integer $n$, $k$ positive integers $0 \leq c_1,\ldots,c_k$, an integer $m \leq k$ and $k$ equations of the form $x_{\alpha_i}+c_i=x_{\beta_i}$, where $\alpha_i,\beta_i\in[n]$ and $i \in [k]$. We require that $\sum_{i=1}^k c_i \leq p(n)$ for some fixed polynomial $n$. \newline
 {\noindent \bf Output:} {\em Yes}, if we can find $mn$ integers  $x_{y,z} \geq 0$ (for each $1\le y\le m$ and $1\le z\le n$) such that for each $i \in [k]$ there exists $1 \leq y \leq m$ such that  $x_{y,\alpha_i}+c_i=x_{y,\beta_i}$ (that is the assignment $x_1,\ldots,x_n = x_{y,1},\ldots,x_{y,n}$ satisfies the $i$\th equation); otherwise {\em No}.

%% file: related.tex
\section{Related Work} \seclab{related}
The sequential black pebbling game was introduced by Hewitt and Paterson~\cite{HP70}, and by Cook~\cite{Coo73}. It has been particularly useful in exploring space/time trade-offs for various problems like matrix multiplication~\cite{Tom78}, fast fourier transformations~\cite{SS78,Tom78}, integer multiplication~\cite{SS79a} and many others~\cite{Cha73,SS79b}. In cryptography it has been used to construct/analyze proofs of space \cite{C:DFKP15,RenD16}, proofs of work~\cite{C:DwoNaoWee05,ITCS:MahMorVad13} and memory-bound functions~\cite{DGN03} (functions that incur many expensive cache-misses~\cite{NDSS:AbadiBW03}). More recently, the black pebbling game has been used to analyze memory hard functions~e.g., \cite{AS15,AlwenB16a,EC:AlwBloPie17,TCC:AlwTac17}.

\subsection{Password Hashing and Memory Hard Functions}
Users often select low-entropy passwords which are vulnerable to offline attacks if an adversary obtains the cryptographic hash of the user's password.  Thus, it is desirable for a password hashing algorithm to involve a function $f(.)$ which is moderately expensive to compute. The goal is to ensure that, even if an adversary obtains the value $(username, f(pwd,salt), salt)$ (where $salt$ is some randomly chosen value), it is prohibitively expensive to evaluate $f(.,salt)$ for millions (billions) of different password guesses. {\sf PBKDF2} (Password Based Key Derivation Function 2) \cite{kaliski2000pkcs} is a popular moderately hard function which iterates the underlying cryptographic hash function many times (e.g., $2^{10}$). Unfortunately,  {\sf PBKDF2} is insufficient to protect against an adversary who can build customized hardware to evaluate the underlying hash function. The cost computing a hash function $H$ 
like SHA256 or MD5 on an Application Specific Integrated Circuit (ASIC) is dramatically smaller than the cost of computing $H$ on traditional hardware \cite{bitcoinBook}. 

\cite{NDSS:AbadiBW03}, observing that cache-misses are more egalitarian than computation, proposed the use of ``memory-bound'' functions for password hashing --- a function which maximizes the number of expensive cache-misses. 
Percival~\cite{Per09} observed that memory costs tend to be stable across different architectures and proposed the use of memory-hard functions (MHFs) for password hashing. 
Presently, there seems to be a consensus that memory hard functions are the `right tool' for constructing moderately expensive functions. 
Indeed, all entrants in the password hashing competition claimed some form of memory hardness~\cite{PHC}. 
As the name suggests, the cost of computing a memory hard function is primarily memory related (storing/retrieving data values). 
Thus, the cost of computing the function cannot be significantly reduced by constructing an ASIC. 
Percival~\cite{Per09} introduced a candidate memory hard function called \scrypt, but \scrypt is potentially vulnerable to side-channel attacks as its computation yields a memory access pattern that is data-dependent (i.e., depends on the secret input/password).  
Due to the recently completed password hashing competition~\cite{PHC} we have many candidate data-independent memory hard functions such as Catena~\cite{forler2013catena} and the winning contestant Argon2i-A~\cite{Argon2-PHC}.\footnote{The specification of Argon2i has changed several times. We use Argon2i-A to refer to the version of Argon2i from the password hashing competition, and we use Argon2i-B to refer to the version that is currently being considered for standardization by the Cryptography Form Research Group (CFRG) of the IRTF\cite{Argon2-IRTF}. }

\subsubsection{iMHFs and Graph Pebbling}
All known candidate iMHFs can be described using a DAG $G$ and a hash function $H$. Graph pebbling is a particularly useful as a tool to analyze the security of an iMHF~\cite{AS15,CBS16,forler2013catena}. A pebbling of $G$ naturally corresponds to an algorithm to compute the iMHF. Alwen and Serbinenko~\cite{AS15} showed that in the pROM model of computation, {\em any} algorithm to compute the iMHF corresponds to a pebbling of $G$.

\subsubsection{Measuring Pebbling Costs} 
In the past, MHF analysis has focused on space-time complexity~\cite{Per09,forler2013catena,AC:BonCorSch16}. For example, the designers of Catena~\cite{forler2013catena} showed that their DAG $G$ had high sequential space-time pebbling cost $\sst(G)$ and Boneh \etal~\cite{AC:BonCorSch16} showed that Argon2i-A and their own iMHF candidate iBH (``balloon hash'') have (sequential) space-time cost $\Omega\big(n^2\big)$. Alwen and Serbinenko~\cite{AS15} observed that these guarantees are insufficient for two reasons: (1) the adversary may be parallel, and (2) the adversary might amortize his costs over multiple iMHF instances (e.g., multiple password guesses). Indeed, there are now multiple known attacks on Catena~\cite{BK15,AS15,AlwenB16a}. Alwen and Blocki~\cite{AlwenB16a,AlwenB16b} gave attacks on Argon2i-A, Argon2i-B, iBH, and Catena with lower than desired amortized space-time cost --- $\pcc(G) \leq O\big(n^{1.8}\big)$ for Argon2i-B, $\pcc(G) \leq \tilde{O}\big(n^{1.75}\big)$ for Argon2i-A and iBH and $\pcc(G) \leq O\big(n^{5/3}\big)$ for Catena. This motivates the need to study cumulative cost $\pcc$ instead of space-time cost since amortized space-time complexity approaches $\pcc$ as the number of iMHF instances being computed increases.
  
  Alwen \etal~\cite{EC:AlwBloPie17} recently constructed a constant indegree graph $G$ with $\pcc(G) =  \Omega\big(\frac{n^2}{\log n}\big)$. From a theoretical standpoint, this is essentially optimal as any constant $\indeg$ DAG has $\pcc = O\big( \frac{n^2 \log \log n}{\log n}\big)$~\cite{AlwenB16a}, but from a practical standpoint the critically important constants terms in the lower bound are not well understood. 

Ren and Devedas~\cite{TCC:RenDev17} recently proposed an alternative metric MHFs called bandwidth hardness. The key distinction between bandwidth hardness and cumulative pebbling cost is that bandwidth hardness attempts to approximate {\em energy costs}, while cumulative pebbling cost attempts to approximate amortized capital costs (i.e., the cost of all of the DRAM chips divided by the number of MHF instances that can be computed before the DRAM chip fails). In this paper we focus on the cumulative pebbling cost metric as we expect amortized capital costs to dominate for sufficiently large $n$. In particular, bandwidth costs scale linearly with the running time $n$ (at best), while cumulative pebbling costs can scale quadratically with $n$.

\subsection{Computational Complexity of Pebbling} \seclab{relatedPebblingComplexity}
The computational complexity of various graph pebbling has been explored previously in different settings~\cite{gilbert1980pebbling,hertel2010pspace}. Gilbert \etal~\cite{gilbert1980pebbling} focused on space-complexity of the black-pebbling game. Here, the goal is to find a  pebbling which minimizes the total number of pebbles on the graph at any point in time (intuitively this corresponds to minimizing the maximum space required during computation of the associated function). Gilbert \etal~\cite{gilbert1980pebbling} showed that this problem is PSPACE complete by reducing from the truly quantified boolean formula (TQBF) problem. 

The optimal (space-minimizing) pebbling of the graphs from the reduction of Gilbert \etal~\cite{gilbert1980pebbling} often require exponential time. By contrast, observe that $\minCC \in NP$ because any DAG $G$ with $n$ nodes this algorithm has a pebbling $P$ with $\cc(P) \leq \ST(P) \leq n^2$. Thus, if we are minimizing $\cc$ or $\ST$ cost, the optimal pebbling of $G$ will trivially never require more than $n^2$ steps. Thus, we need different tools to analyze the computational complexity of the problem of finding a pebbling with low cumulative cost. 

In \appref{SpaceTime}, we show that the optimal pebbling from~\cite{gilbert1980pebbling} does take polynomial time if the TQBF formula only uses existential quantifiers (i.e., if we reduce from 3SAT). Thus, the reduction of Gilbert \etal~\cite{gilbert1980pebbling} can also be extended to show that it is $\NPcomplete$ to check whether a DAG $G$ admits a pebbling $P$ with $\ST(P) \leq k$ for some parameter $k$. The reduction, which simply appends a long chain to the original graph, exploits the fact that if we increase space-usage even temporarily we will dramatically increase $\ST$ cost. However, this reduction does not extend to cumulative cost because the penalty for temporarily placing large number of pebbles can be quite small as we do not keep these pebbles on the graph for a long time. 

%% file: reduction.tex
\section{NP-Hardness of \minCC}
In this section we prove that $\minCC$ is $\NPcomplete$ by reduction from $\BTwoLC$. 
Showing that $\minCC \in \NP$ is straightforward so we will focus on proving that the decision problem is $\NPhard$. We first provide some intuition about the reduction. 

Recall that a \BTwoLC instance consists of $n$ variables $x_1,\ldots,x_n$, and $k$ equations of the form $x_{\alpha_i}+c_i=x_{\beta_i}$, where $\alpha_i,\beta_i\in[n]$, $i\in[k]$, and each $c_i \leq p(n)$ is a positive integer bounded by some polynomial in $n$. The goal is to determine whether there exist $m$ different variable assignments such that each equation is satisfied by {\em at least one} of the $m$ assignments. Formally, the goal is to decide if there exists a set of $m < k$ variable assignments: $x_{y,z} \geq 0$ for each $1\le y\le m$ and $1\le z\le n$ so that for each $i \in [k]$ there exists $y \in [m]$ such that  $x_{y,\alpha_i}+c_i=x_{y,\beta_i}$ --- that is the $i$\th equation $x_{\alpha_i}+c_i=x_{\beta_i}$ is satisfied by the $y$\th variable  assignment  $x_{y,1},\ldots, x_{y,n}$. 
For example, if $k=2$ and the equations are $x_1 + 1 = x_2$ and $x_2 + 2 = x_3$, then $m=1$ suffices to satisfy all the equations. 
On the other hand, if $x_1 + 1 = x_2$ and $x_1 + 2 = x_2$, then we require $m\ge2$ since the equations are no longer independent. 
Observe that for $m=1$, $\BTwoLC$ seeks a single satisfying assignment, whereas for $m\ge k$, each equation can be satisfied by a separate assignment of the variables (specifically, the $i$\th assignment is all zeroes except $x_{\beta_i}=c_i$).

Suppose we are given an instance of \BTwoLC. 
We shall construct a $\minCC$ instance $G_{\BTwoLC}$ in such a way that the optimal pebbling of $G_{\BTwoLC}$ has ``low'' cost if the instance of \BTwoLC is satisfiable and otherwise, has ``high'' cost. The graph $\BTwoLC$ will be constructed from three different types of gadgets: $\tau$ gadgets $C_i^1,\ldots,C_i^\tau$ for each variable $x_i$, a gadget $E_i$ for each equation and a ``$m$-assignments'' gadget $M_i$ for each variable $x_i$. Here $\tau$ is a parameter we shall set to create a gap between the pebbling costs of satisfiable and unsatisfiable instances of \BTwoLC. Each gadget is described in more detail below.

\paragraph{Variable Gadgets} Our first gadget is a chain of length $c=\sum c_i$ so that each node is connected to the previous node, and can only be pebbled if there exists a pebble on the previous node in the previous step, such as in Figure~\ref{fig:chain}. For each variable $x_i$ in our \BTwoLC instance we will add $\tau$ copies of our chain gadget $C_i^1,\ldots,C_i^\tau$. 
Formally, for each $j \in [\tau]$ the chain gadget $C_i^j$ consists of $c$ vertices $v_i^{j,1},\ldots,v_i^{j,c}$ with directed edges  $\left(v_i^{j,z},v_i^{j,z+1}\right)$ for each $z < c$. 
We will later add a gadget to ensure that we must walk a pebble down each of these chains $m$ different times and that in any optimal pebbling $P \in \pPeb_{G_{\BTwoLC}}$ (with $\cc(P) = \pcc\left(G_{\BTwoLC}\right)$) the walks on each chain gadget $C_i^1,\ldots,C_i^\tau$ are synchronized e.g., for each pebbling round $y$ and for each $z \leq c$ we have $v_i^{j,z} \in P_y \leftrightarrow \{v_i^{1,z},\ldots v_i^{\tau,z} \} \subseteq P_y $. 
Intuitively, each {\em time} at which we begin walking a pebble down these chains will correspond to an assignment of the \BTwoLC variable $x_i$. Hence, it suffices to have $c=\sum c_i$ nodes in the chain.

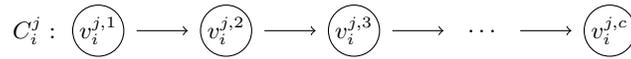
\begin{figure}[H]
\centering
\begin{tikzpicture}[scale=1.7]
\node at (1.5cm,0){$C_i^j:$};
\draw (2cm,0) circle (0.2cm);
\node at (2cm,0){$v_i^{j,1}$};
\draw (3cm,0) circle (0.2cm);
\node at (3cm,0){$v_i^{j,2}$};
\draw (4cm,0) circle (0.2cm);
\node at (4cm,0){$v_i^{j,3}$};
\node at (5cm,0){$\ldots$};
\draw (6cm,0) circle (0.2cm);
\node at (6cm,0){$v_i^{j,c}$};
\draw[->, black] (2.3cm,0cm) -- (2.7cm,0cm);
\draw[->, black] (3.3cm,0cm) -- (3.7cm,0cm);
\draw[->, black] (4.3cm,0cm) -- (4.7cm,0cm);
\draw[->, black] (5.3cm,0cm) -- (5.7cm,0cm);
\end{tikzpicture}
\caption{Example variable gadget $C_i^j$ of length $c=\sum c_i$. $G_{\BTwoLC}$ replicates this gadget $\tau$ times: $C_i^1,\ldots,C_i^\tau$. Each of the $\tau$ copies behaves the same.}\label{fig:chain}
\end{figure}
 

\paragraph{Equation Gadget} 
For the $i$\th equation $x_{\alpha_i}+c_i=x_{\beta_i}$, the gadget $E_i$ is a chain of length $c-c_i$. 
For each $j \in [\tau]$ we connect the equation gadget $E_i$ to each of the variable gadgets $C_{\alpha_i}^j$ and $C_{\beta_i}^j$ as follows: the $a$\th node $e_j^a$ in chain $E_j$ has incoming edges from vertices $v_{\alpha_i}^{l,a}$ and $v_{\beta_i}^{l,a+c_i}$ for all $1\le l\le\tau$, as demonstrated in Figure~\ref{fig:eqn}.
To pebble the equation gadget, the corresponding variable gadgets must be pebbled synchronously, distance $c_i$ apart.

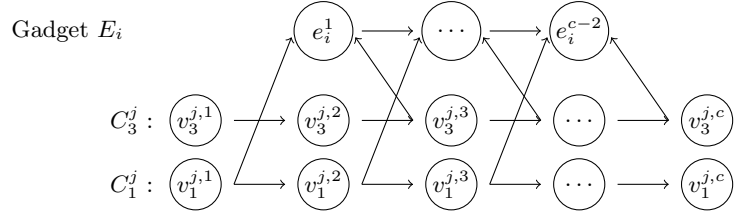
\begin{figure}[H]
\centering
\begin{tikzpicture}[scale=1.7]
\node at (1.5cm,0){$C_{1}^j:$};
\draw (2cm,0) circle (0.2cm);
\node at (2cm,0){$v_{1}^{j,1}$};
\draw (3cm,0) circle (0.2cm);
\node at (3cm,0){$v_{1}^{j,2}$};
\draw (4cm,0) circle (0.2cm);
\node at (4cm,0){$v_{1}^{j,3}$};
\draw (5cm,0) circle (0.2cm);
\node at (5cm,0){$\ldots$};
\draw (6cm,0) circle (0.2cm);
\node at (6cm,0){$v_{1}^{j,c}$};
\draw[->, black] (2.3cm,0cm) -- (2.7cm,0cm);
\draw[->, black] (3.3cm,0cm) -- (3.7cm,0cm);
\draw[->, black] (4.3cm,0cm) -- (4.7cm,0cm);
\draw[->, black] (5.3cm,0cm) -- (5.7cm,0cm);

\node at (1.5cm,0.5cm){$C_{3}^j:$};
\draw (2cm,0.5cm) circle (0.2cm);
\node at (2cm,0.5cm){$v_{3}^{j,1}$};
\draw (3cm,0.5cm) circle (0.2cm);
\node at (3cm,0.5cm){$v_{3}^{j,2}$};
\draw (4cm,0.5cm) circle (0.2cm);
\node at (4cm,0.5cm){$v_{3}^{j,3}$};
\draw (5cm,0.5cm) circle (0.2cm);
\node at (5cm,0.5cm){$\ldots$};
\draw (6cm,0.5cm) circle (0.2cm);
\node at (6cm,0.5cm){$v_{3}^{j,c}$};
\draw[->, black] (2.3cm,0.5cm) -- (2.7cm,0.5cm);
\draw[->, black] (3.3cm,0.5cm) -- (3.7cm,0.5cm);
\draw[->, black] (4.3cm,0.5cm) -- (4.7cm,0.5cm);
\draw[->, black] (5.3cm,0.5cm) -- (5.7cm,0.5cm);

\node at (1cm,1.2cm){Gadget $E_i$};
\draw (3cm,1.2cm) circle (0.23cm);
\node at (3cm,1.2cm){$e_i^1$};
\draw (4cm,1.2cm) circle (0.23cm);
\node at (4cm,1.2cm){$\ldots$};
\draw (5cm,1.2cm) circle (0.23cm);
\node at (5cm,1.2cm){$e_i^{c-2}$};
\draw[->, black] (3.3cm,1.2cm) -- (3.7cm,1.2cm);
\draw[->, black] (4.3cm,1.2cm) -- (4.7cm,1.2cm);
\draw[->, black] (2.3cm,0cm) -- (2.75cm,1.15cm);
\draw[->, black] (3.7cm,0.5cm) -- (3.25cm,1.15cm);
\draw[->, black] (3.3cm,0cm) -- (3.75cm,1.15cm);
\draw[->, black] (4.7cm,0.5cm) -- (4.25cm,1.15cm);
\draw[->, black] (4.3cm,0cm) -- (4.75cm,1.15cm);
\draw[->, black] (5.7cm,0.5cm) -- (5.25cm,1.15cm);
\end{tikzpicture}
\caption{The gadget $E_i$ for equation $x_3+2=x_1$. The example shows how $E_i$ is connected to the variable gadgets $C_1^j$ and $C_3^j$ for each $j \in [\tau]$.}\label{fig:eqn}
\end{figure}

Intuitively, if the equation $x_{\alpha}+c_i=x_{\beta}$ is satisfied by the $j$\th assignment, then on the $j$\th time we walk pebbles down the chain $x_{\alpha}$ and $x_{\beta}$, the pebbles on each chain will be synchronized (i.e., when we have a pebble on $v_{\alpha}^{l,a}$, the $a$\th link in the chain representing $x_{\alpha}$ we will have a pebble on the node $v_{\beta}^{l,a+c_i}$ during the same round) so that we can pebble all of the nodes in the equation gadget, such as in Figure \ref{fig:eqn:pebbling}. On the other hand, if the pebbles on each chain are not synchronized appropriately, we cannot pebble the equation gadget. Finally, we create a single sink node linked from each of the $k$ equation chains, which can only be pebbled if all equation nodes are pebbled. 

\begin{figure}[!htb]
\begin{subfigure}
\centering
\begin{tikzpicture}[scale=1]
\node at (1.2cm,0){$C_{1}^j:$};
\draw (2cm,0) circle (0.2cm);
\draw (3cm,0) circle (0.2cm);
\draw (4cm,0) circle (0.2cm);
\draw (5cm,0) circle (0.2cm);
\draw (6cm,0) circle (0.2cm);
\draw[->, black] (2.3cm,0cm) -- (2.7cm,0cm);
\draw[->, black] (3.3cm,0cm) -- (3.7cm,0cm);
\draw[->, black] (4.3cm,0cm) -- (4.7cm,0cm);
\draw[->, black] (5.3cm,0cm) -- (5.7cm,0cm);

\node at (1.2cm,0.5cm){$C_{3}^j:$};
\filldraw[shading=radial,inner color=white, outer color=gray!75, opacity=0.2] (2cm,0.5cm) circle (0.2cm);
\draw (2cm,0.5cm) circle (0.2cm);
\draw (3cm,0.5cm) circle (0.2cm);
\draw(4cm,0.5cm) circle (0.2cm);
\draw (5cm,0.5cm) circle (0.2cm);
\draw (6cm,0.5cm) circle (0.2cm);
\draw[->, black] (2.3cm,0.5cm) -- (2.7cm,0.5cm);
\draw[->, black] (3.3cm,0.5cm) -- (3.7cm,0.5cm);
\draw[->, black] (4.3cm,0.5cm) -- (4.7cm,0.5cm);
\draw[->, black] (5.3cm,0.5cm) -- (5.7cm,0.5cm);

\node at (1cm,1.2cm){Time step $1$};
\draw (3cm,1.2cm) circle (0.23cm);
\draw (4cm,1.2cm) circle (0.23cm);
\draw (5cm,1.2cm) circle (0.23cm);
\draw[->, black] (3.3cm,1.2cm) -- (3.7cm,1.2cm);
\draw[->, black] (4.3cm,1.2cm) -- (4.7cm,1.2cm);
\draw[->, black] (2.3cm,0cm) -- (2.75cm,1.15cm);
\draw[->, black] (3.7cm,0.5cm) -- (3.25cm,1.15cm);
\draw[->, black] (3.3cm,0cm) -- (3.75cm,1.15cm);
\draw[->, black] (4.7cm,0.5cm) -- (4.25cm,1.15cm);
\draw[->, black] (4.3cm,0cm) -- (4.75cm,1.15cm);
\draw[->, black] (5.7cm,0.5cm) -- (5.25cm,1.15cm);
\end{tikzpicture}
\end{subfigure}
\begin{subfigure}
\centering
\begin{tikzpicture}[scale=1]
\node at (1.2cm,0){$C_{1}^j:$};
\draw (2cm,0) circle (0.2cm);
\draw (3cm,0) circle (0.2cm);
\draw (4cm,0) circle (0.2cm);
\draw (5cm,0) circle (0.2cm);
\draw (6cm,0) circle (0.2cm);
\draw[->, black] (2.3cm,0cm) -- (2.7cm,0cm);
\draw[->, black] (3.3cm,0cm) -- (3.7cm,0cm);
\draw[->, black] (4.3cm,0cm) -- (4.7cm,0cm);
\draw[->, black] (5.3cm,0cm) -- (5.7cm,0cm);

\node at (1.2cm,0.5cm){$C_{3}^j:$};
\draw (2cm,0.5cm) circle (0.2cm);
\filldraw[shading=radial,inner color=white, outer color=gray!75, opacity=0.2](3cm,0.5cm) circle (0.2cm);
\draw (3cm,0.5cm) circle (0.2cm);
\draw (4cm,0.5cm) circle (0.2cm);
\draw (5cm,0.5cm) circle (0.2cm);
\draw (6cm,0.5cm) circle (0.2cm);
\draw[->, black] (2.3cm,0.5cm) -- (2.7cm,0.5cm);
\draw[->, black] (3.3cm,0.5cm) -- (3.7cm,0.5cm);
\draw[->, black] (4.3cm,0.5cm) -- (4.7cm,0.5cm);
\draw[->, black] (5.3cm,0.5cm) -- (5.7cm,0.5cm);

\node at (1cm,1.2cm){Time step $2$};
\draw (3cm,1.2cm) circle (0.23cm);
\draw (4cm,1.2cm) circle (0.23cm);
\draw (5cm,1.2cm) circle (0.23cm);
\draw[->, black] (3.3cm,1.2cm) -- (3.7cm,1.2cm);
\draw[->, black] (4.3cm,1.2cm) -- (4.7cm,1.2cm);
\draw[->, black] (2.3cm,0cm) -- (2.75cm,1.15cm);
\draw[->, black] (3.7cm,0.5cm) -- (3.25cm,1.15cm);
\draw[->, black] (3.3cm,0cm) -- (3.75cm,1.15cm);
\draw[->, black] (4.7cm,0.5cm) -- (4.25cm,1.15cm);
\draw[->, black] (4.3cm,0cm) -- (4.75cm,1.15cm);
\draw[->, black] (5.7cm,0.5cm) -- (5.25cm,1.15cm);
\end{tikzpicture}
\end{subfigure}

\begin{subfigure}
\centering
\begin{tikzpicture}[scale=1]
\node at (1.2cm,0){$C_{1}^j:$};
\filldraw[shading=radial,inner color=white, outer color=gray!75, opacity=0.2](2cm,0) circle (0.2cm);
\draw (2cm,0) circle (0.2cm);
\draw (3cm,0) circle (0.2cm);
\draw (4cm,0) circle (0.2cm);
\draw (5cm,0) circle (0.2cm);
\draw (6cm,0) circle (0.2cm);
\draw[->, black] (2.3cm,0cm) -- (2.7cm,0cm);
\draw[->, black] (3.3cm,0cm) -- (3.7cm,0cm);
\draw[->, black] (4.3cm,0cm) -- (4.7cm,0cm);
\draw[->, black] (5.3cm,0cm) -- (5.7cm,0cm);

\node at (1.2cm,0.5cm){$C_{3}^j:$};
\draw (2cm,0.5cm) circle (0.2cm);
\draw (3cm,0.5cm) circle (0.2cm);
\filldraw[shading=radial,inner color=white, outer color=gray!75, opacity=0.2](4cm,0.5cm) circle (0.2cm);
\draw (4cm,0.5cm) circle (0.2cm);
\draw (5cm,0.5cm) circle (0.2cm);
\draw (6cm,0.5cm) circle (0.2cm);
\draw[->, black] (2.3cm,0.5cm) -- (2.7cm,0.5cm);
\draw[->, black] (3.3cm,0.5cm) -- (3.7cm,0.5cm);
\draw[->, black] (4.3cm,0.5cm) -- (4.7cm,0.5cm);
\draw[->, black] (5.3cm,0.5cm) -- (5.7cm,0.5cm);

\node at (1cm,1.2cm){Time step $3$};
\draw (3cm,1.2cm) circle (0.23cm);
\draw (4cm,1.2cm) circle (0.23cm);
\draw (5cm,1.2cm) circle (0.23cm);
\draw[->, black] (3.3cm,1.2cm) -- (3.7cm,1.2cm);
\draw[->, black] (4.3cm,1.2cm) -- (4.7cm,1.2cm);
\draw[->, black] (2.3cm,0cm) -- (2.75cm,1.15cm);
\draw[->, black] (3.7cm,0.5cm) -- (3.25cm,1.15cm);
\draw[->, black] (3.3cm,0cm) -- (3.75cm,1.15cm);
\draw[->, black] (4.7cm,0.5cm) -- (4.25cm,1.15cm);
\draw[->, black] (4.3cm,0cm) -- (4.75cm,1.15cm);
\draw[->, black] (5.7cm,0.5cm) -- (5.25cm,1.15cm);
\end{tikzpicture}
\end{subfigure}
\begin{subfigure}
\centering
\begin{tikzpicture}[scale=1]
\node at (1.2cm,0){$C_{1}^j:$};
\draw (2cm,0) circle (0.2cm);
\filldraw[shading=radial,inner color=white, outer color=gray!75, opacity=0.2] (3cm,0) circle (0.2cm);
\draw(3cm,0) circle (0.2cm);
\draw (4cm,0) circle (0.2cm);
\draw (5cm,0) circle (0.2cm);
\draw (6cm,0) circle (0.2cm);
\draw[->, black] (2.3cm,0cm) -- (2.7cm,0cm);
\draw[->, black] (3.3cm,0cm) -- (3.7cm,0cm);
\draw[->, black] (4.3cm,0cm) -- (4.7cm,0cm);
\draw[->, black] (5.3cm,0cm) -- (5.7cm,0cm);

\node at (1.2cm,0.5cm){$C_{3}^j:$};
\draw (2cm,0.5cm) circle (0.2cm);
\draw (3cm,0.5cm) circle (0.2cm);
\draw (4cm,0.5cm) circle (0.2cm);
\filldraw[shading=radial,inner color=white, outer color=gray!75, opacity=0.2](5cm,0.5cm) circle (0.2cm);
\draw(5cm,0.5cm) circle (0.2cm);
\draw (6cm,0.5cm) circle (0.2cm);
\draw[->, black] (2.3cm,0.5cm) -- (2.7cm,0.5cm);
\draw[->, black] (3.3cm,0.5cm) -- (3.7cm,0.5cm);
\draw[->, black] (4.3cm,0.5cm) -- (4.7cm,0.5cm);
\draw[->, black] (5.3cm,0.5cm) -- (5.7cm,0.5cm);

\node at (1cm,1.2cm){Time step $4$};
\filldraw[shading=radial,inner color=white, outer color=gray!75, opacity=0.2](3cm,1.2cm) circle (0.23cm);
\draw(3cm,1.2cm) circle (0.23cm);
\draw (4cm,1.2cm) circle (0.23cm);
\draw (5cm,1.2cm) circle (0.23cm);
\draw[->, black] (3.3cm,1.2cm) -- (3.7cm,1.2cm);
\draw[->, black] (4.3cm,1.2cm) -- (4.7cm,1.2cm);
\draw[->, black] (2.3cm,0cm) -- (2.75cm,1.15cm);
\draw[->, black] (3.7cm,0.5cm) -- (3.25cm,1.15cm);
\draw[->, black] (3.3cm,0cm) -- (3.75cm,1.15cm);
\draw[->, black] (4.7cm,0.5cm) -- (4.25cm,1.15cm);
\draw[->, black] (4.3cm,0cm) -- (4.75cm,1.15cm);
\draw[->, black] (5.7cm,0.5cm) -- (5.25cm,1.15cm);
\end{tikzpicture}
\end{subfigure}

\begin{subfigure}
\centering
\begin{tikzpicture}[scale=1]
\node at (1.2cm,0){$C_{1}^j:$};
\draw (2cm,0) circle (0.2cm);
\draw (3cm,0) circle (0.2cm);
\filldraw[shading=radial,inner color=white, outer color=gray!75, opacity=0.2](4cm,0) circle (0.2cm);
\draw(4cm,0) circle (0.2cm);
\draw (5cm,0) circle (0.2cm);
\draw (6cm,0) circle (0.2cm);
\draw[->, black] (2.3cm,0cm) -- (2.7cm,0cm);
\draw[->, black] (3.3cm,0cm) -- (3.7cm,0cm);
\draw[->, black] (4.3cm,0cm) -- (4.7cm,0cm);
\draw[->, black] (5.3cm,0cm) -- (5.7cm,0cm);

\node at (1.2cm,0.5cm){$C_{3}^j:$};
\draw (2cm,0.5cm) circle (0.2cm);
\draw (3cm,0.5cm) circle (0.2cm);
\draw (4cm,0.5cm) circle (0.2cm);
\draw (5cm,0.5cm) circle (0.2cm);
\filldraw[shading=radial,inner color=white, outer color=gray!75, opacity=0.2](6cm,0.5cm) circle (0.2cm);
\draw(6cm,0.5cm) circle (0.2cm);
\draw[->, black] (2.3cm,0.5cm) -- (2.7cm,0.5cm);
\draw[->, black] (3.3cm,0.5cm) -- (3.7cm,0.5cm);
\draw[->, black] (4.3cm,0.5cm) -- (4.7cm,0.5cm);
\draw[->, black] (5.3cm,0.5cm) -- (5.7cm,0.5cm);

\node at (1cm,1.2cm){Time step $5$};
\draw (3cm,1.2cm) circle (0.23cm);
\filldraw[shading=radial,inner color=white, outer color=gray!75, opacity=0.2](4cm,1.2cm) circle (0.23cm);
\draw(4cm,1.2cm) circle (0.23cm);
\draw (5cm,1.2cm) circle (0.23cm);
\draw[->, black] (3.3cm,1.2cm) -- (3.7cm,1.2cm);
\draw[->, black] (4.3cm,1.2cm) -- (4.7cm,1.2cm);
\draw[->, black] (2.3cm,0cm) -- (2.75cm,1.15cm);
\draw[->, black] (3.7cm,0.5cm) -- (3.25cm,1.15cm);
\draw[->, black] (3.3cm,0cm) -- (3.75cm,1.15cm);
\draw[->, black] (4.7cm,0.5cm) -- (4.25cm,1.15cm);
\draw[->, black] (4.3cm,0cm) -- (4.75cm,1.15cm);
\draw[->, black] (5.7cm,0.5cm) -- (5.25cm,1.15cm);
\end{tikzpicture}
\end{subfigure}
\begin{subfigure}
\centering
\begin{tikzpicture}[scale=1]
\node at (1.2cm,0){$C_{1}^j:$};
\draw (2cm,0) circle (0.2cm);
\draw (3cm,0) circle (0.2cm);
\draw (4cm,0) circle (0.2cm);
\filldraw[shading=radial,inner color=white, outer color=gray!75, opacity=0.2](5cm,0) circle (0.2cm);
\draw (5cm,0) circle (0.2cm);
\draw (6cm,0) circle (0.2cm);
\draw[->, black] (2.3cm,0cm) -- (2.7cm,0cm);
\draw[->, black] (3.3cm,0cm) -- (3.7cm,0cm);
\draw[->, black] (4.3cm,0cm) -- (4.7cm,0cm);
\draw[->, black] (5.3cm,0cm) -- (5.7cm,0cm);

\node at (1.2cm,0.5cm){$C_{3}^j:$};
\draw (2cm,0.5cm) circle (0.2cm);
\draw (3cm,0.5cm) circle (0.2cm);
\draw (4cm,0.5cm) circle (0.2cm);
\draw (5cm,0.5cm) circle (0.2cm);
\draw (6cm,0.5cm) circle (0.2cm);
\draw[->, black] (2.3cm,0.5cm) -- (2.7cm,0.5cm);
\draw[->, black] (3.3cm,0.5cm) -- (3.7cm,0.5cm);
\draw[->, black] (4.3cm,0.5cm) -- (4.7cm,0.5cm);
\draw[->, black] (5.3cm,0.5cm) -- (5.7cm,0.5cm);

\node at (1cm,1.2cm){Time step $6$};
\draw (3cm,1.2cm) circle (0.23cm);
\draw (4cm,1.2cm) circle (0.23cm);
\filldraw[shading=radial,inner color=white, outer color=gray!75, opacity=0.2](5cm,1.2cm) circle (0.23cm);
\draw(5cm,1.2cm) circle (0.23cm);
\draw[->, black] (3.3cm,1.2cm) -- (3.7cm,1.2cm);
\draw[->, black] (4.3cm,1.2cm) -- (4.7cm,1.2cm);
\draw[->, black] (2.3cm,0cm) -- (2.75cm,1.15cm);
\draw[->, black] (3.7cm,0.5cm) -- (3.25cm,1.15cm);
\draw[->, black] (3.3cm,0cm) -- (3.75cm,1.15cm);
\draw[->, black] (4.7cm,0.5cm) -- (4.25cm,1.15cm);
\draw[->, black] (4.3cm,0cm) -- (4.75cm,1.15cm);
\draw[->, black] (5.7cm,0.5cm) -- (5.25cm,1.15cm);
\end{tikzpicture}
\end{subfigure}
\caption{A pebbling of the equation gadget $x_3+2=x_1$ (at the top) using the satisfying assignment $x_3=1$ and $x_1=3$.}\label{fig:eqn:pebbling}
\end{figure}
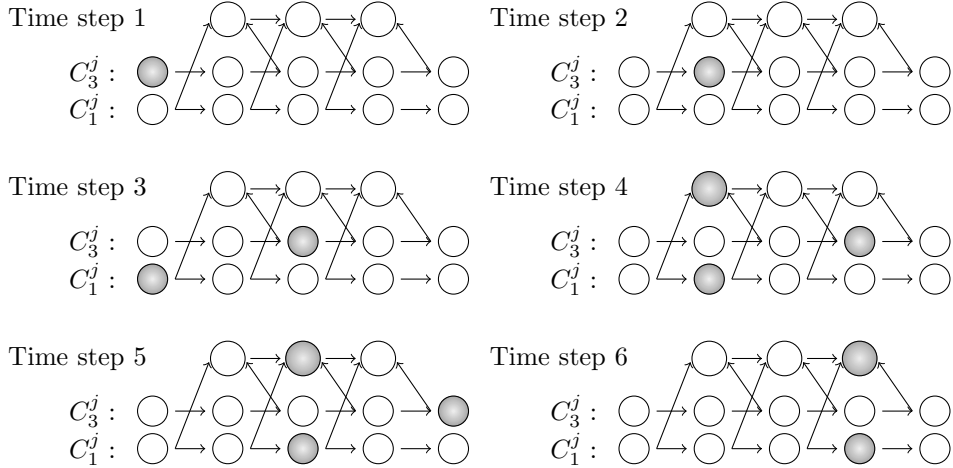

We will use another gadget, the \emph{assignment gadget}, to ensure that in any legal pebbling, we need to ``walk'' a pebble down each chain $C_i^j$ on $m$ different times. 
Each node $v_i^{j,z}$ of a variable gadget in a satisfiable \BTwoLC instance has a pebble on it during exactly $m$ rounds. 
On the other hand, the assignment gadget ensures that for any unsatisfiable \BTwoLC instance, there exists some $i \leq n$ and $z \leq c$ such that each of the nodes $v_i^{1,z},\ldots,v_i^{\tau,z}$ are pebbled during at least $m+1$ rounds.

We will tune the parameter $\tau$ to ensure that any such pebbling is more expensive, formalized in \claimref{claim:sync} in \appref{MissingProofs}.

\paragraph{$m$ assignments gadget} Our final gadget is a path of length $cm$ so that each node is connected to the previous node. We create a path gadget $M_i$ of length $cm$ for each variable $x_i$ and connect $M_i$ to each the variable gadgets $C_i^1,\ldots C_i^\tau$ as follows: for every node $z_{i}^{p+qc}$ in the path with position $p+qc>1$, where $1\le p\le c$ and $0\le q<m-1$, we add an edge to $z_{i}^{p+qc}$ from each of the nodes $v_i^{j,p}$, $1\le j\le\tau$ (that is, the $p$\th node in all $\tau$ chains $C_i^1,\ldots C_i^\tau$ representing the variable $x_i$ ). 
We connect the final node in each of the $n$ paths to the final sink node $v_{sink}$ in our graph $G_{\BTwoLC}$.

Intuitively, to pebble $v_{sink}$ we must walk a pebble down each of the gadgets $M_i$ which in turn requires us to walk a pebble along each chain $C_i^j$, $1\le j\le\tau$, at least $m$ times.
For example, see Figure~\ref{fig:path}.

 \begin{figure}[H]
\centering
\begin{tikzpicture}[scale=1.4]
\draw (2cm,0) circle (0.2cm);
\node at (2cm,0){$z_i^1$};
\draw (3cm,0) circle (0.2cm);
\node at (3cm,0){$\ldots$};
\draw (4cm,0) circle (0.2cm);
\node at (4cm,0){$z_i^{c}$};
\draw (5cm,0) circle (0.2cm);
\node at (5cm,0){$\ldots$};
\draw (6cm,0) circle (0.2cm);
\node at (6cm,0){$\ldots$};
\draw (7cm,0) circle (0.2cm);
\node at (7cm,0){$z_i^{2c}$};
\draw (8cm,0) circle (0.2cm);
\node at (8cm,0){$\ldots$};
\draw (9cm,0) circle (0.2cm);
\node at (9cm,0){$\ldots$};
\draw (10cm,0) circle (0.2cm);
\node at (10cm,0){$z_i^{cm}$};

\draw[->, black] (2.3cm,0cm) -- (2.7cm,0cm);
\draw[->, black] (3.3cm,0cm) -- (3.7cm,0cm);
\draw[->, black] (4.3cm,0cm) -- (4.7cm,0cm);
\draw[->, black] (5.3cm,0cm) -- (5.7cm,0cm);
\draw[->, black] (6.3cm,0cm) -- (6.7cm,0cm);
\draw[->, black] (7.3cm,0cm) -- (7.7cm,0cm);
\draw[->, black] (8.3cm,0cm) -- (8.7cm,0cm);
\draw[->, black] (9.3cm,0cm) -- (9.7cm,0cm);

\node at (2.9cm,0.8cm){Chain for $C_i^j:$};
\node at (5cm,0.7cm){$v_i^{j,1}$};
\draw (5cm,0.7cm) circle (0.2cm);
\node at (6cm,0.7cm){$\ldots$};
\draw (6cm,0.7cm) circle (0.2cm);
\node at (7cm,0.7cm){$v_i^{j,c}$};
\draw (7cm,0.7cm) circle (0.2cm);
\draw[->, black] (5.3cm,0.7cm) -- (5.7cm,0.7cm);
\draw[->, black] (6.3cm,0.7cm) -- (6.7cm,0.7cm);

\node at (1.8cm,0.4cm){$M_i$:};
\draw[->, black] (4.8cm,0.7cm) -- (2.25cm,0.05cm);
\draw[->, black] (5.8cm,0.7cm) -- (3.25cm,0.06cm);
\draw[->, black] (6.8cm,0.7cm) -- (4.25cm,0.07cm);
\draw[->, black] (5cm,0.55cm) -- (5cm,0.18cm);
\draw[->, black] (6cm,0.55cm) -- (6cm,0.18cm);
\draw[->, black] (7cm,0.55cm) -- (7cm,0.18cm);
\draw[->, black] (5.2cm,0.7cm) -- (7.75cm,0.065cm);
\draw[->, black] (6.2cm,0.7cm) -- (8.75cm,0.065cm);
\draw[->, black] (7.2cm,0.7cm) -- (9.65cm,0.065cm);
\end{tikzpicture}
\caption{The gadget $M_i$  for variable $x_i$ is a path of length $cm$. The example shows  how $M_i$ is connected to  $C_i^j$ for each $j \in [\tau]$. The example shows $m=3$ passes and $c=3$.}\label{fig:path}
\end{figure}
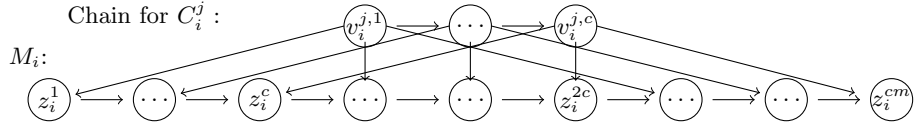
Figure~\ref{fig:all} shows an example of a reduction in its entirety when $\tau=1$.

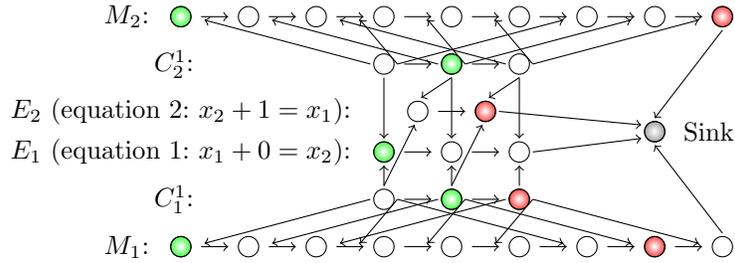
\begin{figure}[htb]
\centering
\begin{tikzpicture}[scale=0.9]
\node at (1.2cm,3.4cm){$M_2$:};
\node at (1.2cm,0.0cm){$M_1$:};
\filldraw[shading=radial,inner color=white, outer color=green!75, opacity=0.8](2cm,0cm) circle (0.15cm);
\draw (2cm,0) circle (0.15cm);
\draw (3cm,0) circle (0.15cm);
\draw (4cm,0) circle (0.15cm);
\draw (5cm,0) circle (0.15cm);
\draw (6cm,0) circle (0.15cm);
\draw (7cm,0) circle (0.15cm);
\draw (8cm,0) circle (0.15cm);
\filldraw[shading=radial,inner color=white, outer color=red!75, opacity=0.8](9cm,0cm) circle (0.15cm);
\draw (9cm,0) circle (0.15cm);
\draw (10cm,0) circle (0.15cm);

\draw[->, black] (2.3cm,0cm) -- (2.7cm,0cm);
\draw[->, black] (3.3cm,0cm) -- (3.7cm,0cm);
\draw[->, black] (4.3cm,0cm) -- (4.7cm,0cm);
\draw[->, black] (5.3cm,0cm) -- (5.7cm,0cm);
\draw[->, black] (6.3cm,0cm) -- (6.7cm,0cm);
\draw[->, black] (7.3cm,0cm) -- (7.7cm,0cm);
\draw[->, black] (8.3cm,0cm) -- (8.7cm,0cm);
\draw[->, black] (9.3cm,0cm) -- (9.7cm,0cm);

\node at (1.9cm,0.7cm){$C_1^1$:};
\draw (5cm,0.7cm) circle (0.15cm);
\filldraw[shading=radial,inner color=white, outer color=green!75, opacity=0.8](6cm,0.7cm) circle (0.15cm);
\draw (6cm,0.7cm) circle (0.15cm);
\filldraw[shading=radial,inner color=white, outer color=red!75, opacity=0.8](7cm,0.7cm) circle (0.15cm);
\draw (7cm,0.7cm) circle (0.15cm);
\draw[->, black] (5.3cm,0.7cm) -- (5.7cm,0.7cm);
\draw[->, black] (6.3cm,0.7cm) -- (6.7cm,0.7cm);

\draw[->, black] (4.8cm,0.7cm) -- (2.35cm,0.05cm);
\draw[->, black] (5.8cm,0.7cm) -- (3.35cm,0.05cm);
\draw[->, black] (6.8cm,0.7cm) -- (4.35cm,0.05cm);
\draw[->, black] (5.2cm,0.7cm) -- (4.65cm,0.05cm);
\draw[->, black] (6.2cm,0.7cm) -- (5.65cm,0.05cm);
\draw[->, black] (7.2cm,0.7cm) -- (6.65cm,0.05cm);
\draw[->, black] (5.2cm,0.7cm) -- (7.65cm,0.05cm);
\draw[->, black] (6.2cm,0.7cm) -- (8.65cm,0.05cm);
\draw[->, black] (7.2cm,0.7cm) -- (9.65cm,0.05cm);

\filldraw[shading=radial,inner color=white, outer color=green!75, opacity=0.8](2cm,3.4cm) circle (0.15cm);
\draw (2cm,3.4cm) circle (0.15cm);
\draw (3cm,3.4cm) circle (0.15cm);
\draw (4cm,3.4cm) circle (0.15cm);
\draw (5cm,3.4cm) circle (0.15cm);
\draw (6cm,3.4cm) circle (0.15cm);
\draw (7cm,3.4cm) circle (0.15cm);
\draw (8cm,3.4cm) circle (0.15cm);
\draw (9cm,3.4cm) circle (0.15cm);
\filldraw[shading=radial,inner color=white, outer color=red!75, opacity=0.8](10cm,3.4cm) circle (0.15cm);
\draw (10cm,3.4cm) circle (0.15cm);

\draw[->, black] (2.3cm,3.4cm) -- (2.7cm,3.4cm);
\draw[->, black] (3.3cm,3.4cm) -- (3.7cm,3.4cm);
\draw[->, black] (4.3cm,3.4cm) -- (4.7cm,3.4cm);
\draw[->, black] (5.3cm,3.4cm) -- (5.7cm,3.4cm);
\draw[->, black] (6.3cm,3.4cm) -- (6.7cm,3.4cm);
\draw[->, black] (7.3cm,3.4cm) -- (7.7cm,3.4cm);
\draw[->, black] (8.3cm,3.4cm) -- (8.7cm,3.4cm);
\draw[->, black] (9.3cm,3.4cm) -- (9.7cm,3.4cm);

\node at (1.9cm,2.7cm){$C_2^1$:};
\draw (5cm,2.7cm) circle (0.15cm);
\filldraw[shading=radial,inner color=white, outer color=green!75, opacity=0.8](6cm,2.7cm) circle (0.15cm);
\draw (6cm,2.7cm) circle (0.15cm);
\draw (7cm,2.7cm) circle (0.15cm);
\draw[->, black] (5.3cm,2.7cm) -- (5.7cm,2.7cm);
\draw[->, black] (6.3cm,2.7cm) -- (6.7cm,2.7cm);

\draw[->, black] (4.8cm,2.7cm) -- (2.35cm,3.35cm);
\draw[->, black] (5.8cm,2.7cm) -- (3.35cm,3.35cm);
\draw[->, black] (6.8cm,2.7cm) -- (4.35cm,3.35cm);
\draw[->, black] (5.2cm,2.7cm) -- (4.65cm,3.35cm);
\draw[->, black] (6.2cm,2.7cm) -- (5.65cm,3.35cm);
\draw[->, black] (7.2cm,2.7cm) -- (6.65cm,3.35cm);
\draw[->, black] (5.2cm,2.7cm) -- (7.65cm,3.35cm);
\draw[->, black] (6.2cm,2.7cm) -- (8.65cm,3.35cm);
\draw[->, black] (7.2cm,2.7cm) -- (9.65cm,3.35cm);

\node at (2cm,1.4cm){$E_1$ (equation 1: $x_1+0=x_2$):};
\filldraw[shading=radial,inner color=white, outer color=green!75, opacity=0.8](5cm,1.4cm) circle (0.15cm);
\draw (5cm,1.4cm) circle (0.15cm);
\draw (6cm,1.4cm) circle (0.15cm);
\draw (7cm,1.4cm) circle (0.15cm);
\draw[->, black] (5.3cm,1.4cm) -- (5.7cm,1.4cm);
\draw[->, black] (6.3cm,1.4cm) -- (6.7cm,1.4cm);

\draw[->, black] (5cm,0.9cm) -- (5cm,1.2cm);
\draw[->, black] (5cm,2.5cm) -- (5cm,1.6cm);
\draw[->, black] (6cm,0.9cm) -- (6cm,1.2cm);
\draw[->, black] (6cm,2.5cm) -- (6cm,1.6cm);
\draw[->, black] (7cm,0.9cm) -- (7cm,1.2cm);
\draw[->, black] (7cm,2.5cm) -- (7cm,1.6cm);

\node at (2cm,2cm){$E_2$ (equation 2: $x_2+1=x_1$):};
\draw (5.5cm,2cm) circle (0.15cm);
\filldraw[shading=radial,inner color=white, outer color=red!75, opacity=0.8] (6.5cm,2cm) circle (0.15cm);
\draw (6.5cm,2cm) circle (0.15cm);
\draw[->, black] (5.8cm,2cm) -- (6.2cm,2cm);

\draw[->, black] (5cm,0.9cm) -- (5.45cm,1.8cm);
\draw[->, black] (6cm,2.5cm) -- (5.55cm,2.2cm);
\draw[->, black] (6cm,0.9cm) -- (6.45cm,1.8cm);
\draw[->, black] (7cm,2.5cm) -- (6.55cm,2.2cm);

\filldraw[shading=radial,inner color=white, outer color=gray!75, opacity=0.8] (9cm,1.7cm) circle (0.15cm);
\draw (9cm,1.7cm) circle (0.15cm);
\node at (9.8cm,1.7cm){Sink};

\draw[->, black] (10cm,0.2cm) -- (9cm,1.5cm);
\draw[->, black] (10cm,3.2cm) -- (9cm,1.9cm);
\draw[->, black] (6.7cm,2cm) -- (8.8cm,1.8cm);
\draw[->, black] (7.2cm,1.4cm) -- (8.8cm,1.6cm);
\end{tikzpicture}
\caption{An example of a complete reduction $G_{\BTwoLC}$, again $m=3$ and $c=3$. The green nodes represent the pebbled vertices at time step $2$ while the red nodes represent the pebbled vertices at time step $10$.}\label{fig:all}
\end{figure}

\newcommand{\LemmaExistSolution}{If the \BTwoLC instance has a valid solution, then $\pcc\big(G_{\BTwoLC}\big)\le\tau cmn + 2cmn + 2ckm + 1$.}
\begin{lemma}
\lemmlab{exists:solution}
\LemmaExistSolution
\end{lemma}

\newcommand{\LemmaNoSolution}{If  the \BTwoLC instance does not have a valid solution, then $\pcc\big(G_{\BTwoLC}\big)\ge\tau cmn + \tau$.}
\begin{lemma}
\lemmlab{no:solution}
\LemmaNoSolution
\end{lemma}

 We outline the key intuition behind \lemmref{exists:solution} and \lemmref{no:solution} and refer to the appendix for the formal proofs. Intuitively, any solution to \BTwoLC corresponds to $m$ walks across the $\tau n$ chains $C_i^j$, $1\le i\le n$, $1\le j\le\tau$ of length $c$. If the \BTwoLC instance is satisfiable then we can synchronize each of these walks so that we can pebble every equation chain $E_j$ and path $M_j$ along the way. Thus, the total cost is $\tau cmn$ plus the cost to pebble the $k$ equation chains $E_j$ ($\leq 2ckm$), the cost to pebble the $n$ paths $M_j$ ($\leq 2cmn$) plus the cost to pebble the sink node $1$. 
 
We then prove a structural property about the optimal pebbling $P=\left(P_0,\ldots,P_t\right) \in \pPeb_{G_{\BTwoLC}}$. In particular, \claimref{claim:sync} from the appendix states that if $P=\left(P_0,\ldots,P_t\right)$ is optimal (i.e., $\cc(P)=\pcc\left(G_{\BTwoLC}\right)$) then during each pebbling round $y \leq t$ the pebbles on each of the chains $C_i^1,\ldots,C_i^\tau$ are synchronized. Formally, for every $y \leq t$, $i \leq n$ and $z \leq c$ we either have (1) $\left\{v_i^{1,z},\ldots,v_i^{\tau,z} \right\} \subseteq P_y$, or (2) $\left\{v_i^{1,z},\ldots,v_i^{\tau,z} \right\} \bigcap P_y = \emptyset$ --- otherwise we could reduce our pebbling cost by discarding these unnecessary pebbles.
 
 If the \BTwoLC instance is not satisfied then we must adopt a ``cheating'' pebbling strategy $P$, which does not correspond to a \BTwoLC solution. We say that $P$ is a ``cheating'' pebbling if some node $v_i^{j,z} \in C_i^j$ is pebbled during at least $m+1$ rounds. We can use \claimref{claim:sync} to show that the cost of {\em any} ``cheating'' pebbling is at least $\cc(P) \geq \tau \left(mc+1\right)$. In particular, $P$ must incur cost at least $\tau mc (n-1)$ to walk a pebble down each of the chains $C_{i'}^j$ with $i' \neq i$ and $1\le j\le\tau$. By \claimref{claim:sync}, any cheating pebbling $P$ incurs costs at least $\tau(mc+1)$ on each of the chains    $C_i^1,\ldots,C_i^\tau$. Thus, the cumulative cost is at least $\tau cmn + \tau$.
   
\begin{theorem} \thmlab{mainResult}
\minCC is \NPcomplete.
\end{theorem}
\begin{proof}
It suffices to show that there is a polynomial time reduction from \BTwoLC to \minCC since \BTwoLC is \NPcomplete (see \thmref{reduce:3part}). Given an instance $\mathcal{P}$ of \BTwoLC, we create the corresponding graph $G$ as described above. This is clearly achieved in polynomial time. 
By \lemmref{exists:solution}, if $\mathcal{P}$ has a valid solution, then $\pcc(G)\le\tau cmn + 2cmn + 2ckm + 1$.
On the other hand, by \lemmref{no:solution}, if $\mathcal{P}$ does not have a valid solution, then $\pcc(G) \ge\tau cmn + \tau$.
Therefore, setting $\tau > 2cmn + 2ckm + 1$ (such as $\tau = 2cmn + 2ckm + 2$) allows one to solve \BTwoLC given an algorithm which outputs $\pcc(G)$.
\end{proof}

\newcommand{\thmminCCd}{$\minCC_\delta$ is \NPcomplete for each $\delta \geq 2$.}
\begin{theorem} \thmlab{minCCd}
\thmminCCd
\end{theorem}
Note that the only possible nodes in $G_{\BTwoLC}$ with indegree greater than two are the nodes in the equation gadgets $E_1,\ldots,E_m$, and the final sink node. 
The equation gadgets can have indegree up to $2\tau+1$, while the final sink node has indegree $n+m$. To show that $\minCC_\delta$ is \NPcomplete when $\delta=2$ we can replace the incoming edges to each of these nodes with a binary tree, so that all vertices have indegree at most two. 
By changing $\tau$ appropriately, we can still distinguish between instances of $\BTwoLC$ using the output of $\minCC_\delta$.
We refer to the appendix for a sketch of the proof of \thmref{minCCd}.

\newcommand{\thmreduceThreePart}{\BTwoLC is \NPcomplete.}
\begin{theorem} \thmlab{reduce:3part}
\thmreduceThreePart
\end{theorem}

To show that $\BTwoLC$ is $\NPcomplete$ we will reduce from the problem \ThreePartition, which is known to be \NPcomplete. The decision problem \ThreePartition is defined as follows: \\
  {\noindent \bf Input:} A multi-set $S$ of $m=3n$ positive integers $x_1,\ldots, x_m \geq 1$ such that (1) we have $\frac{T}{4n} < x_i < \frac{T}{2n}$ for each $1 \leq i \leq m$, where $T=x_1 + \ldots + x_m$, and (2) we require that $T \leq p(n)$ for a fixed polynomial $p$.\footnote{We may assume $\frac{T}{4n} < x_i < \frac{T}{2n}$ by taking any set of positive integers and adding a large fixed constant to all terms, as described in \cite{Demaine6890}.} \newline
  {\noindent \bf Output:} {\em Yes,} if there is a partition of $[m]$ into $n$ subsets $S_1,\ldots, S_n$ such that $\sum_{j \in S_i} x_j  = \frac{T}{n}$ for each $1 \leq i \leq n$; otherwise {\em No}.

  \begin{fact}\cite{garey1975complexity,Demaine6890} \ThreePartition is \NPcomplete.\footnote{The   \ThreePartition problem is called P[3,1] in \cite{garey1975complexity}.}
  \end{fact}
We refer to the appendix for the proof of \thmref{reduce:3part}, where we show that there is a polynomial time reduction from \ThreePartition to \BTwoLC.

%% file: reducible.tex
\section{NP-Hardness of $\REDUCIBLE_d$}
\seclab{Reducible}
The attacks of Alwen and Blocki~\cite{AlwenB16a,AlwenB16b} exploited the fact that the Argon2i-A, Argon2i-B, iBH and Catena DAGs are not depth-robust. In general, Alwen and Blocki~\cite{AlwenB16a} showed that any $(e,d)$-reducible DAG $G$ can be pebbled with cumulative cost $O\big(ne+n\sqrt{nd}\big)$. Thus, depth-robustness is a necessary condition for a secure iMHF. Recently, Alwen \etal~\cite{EC:AlwBloPie17} showed that depth-robustness is sufficient for a secure iMHF. In particular, they showed that an $(e,d)$-depth reducible graph has $\pcc(G) \geq ed$.\footnote{Alwen \etal~\cite{EC:AlwBloPie17} also gave tighter upper and lower bounds on $\pcc(G)$ for the Argon2i-A, iBH and Catena iMHFs. For example, $\pcc(G) = \Omega\left(n^{1.66} \right)$ and $\pcc(G) = O\left(n^{1.71} \right)$ for a random Argon2i-A DAG $G$ (with high probability). Blocki and Zhou~\cite{TCC:BloZho17} recently tightened the upper and lower bounds on Argon2i-B showing that $\pcc(G)=O\left(n^{1.767} \right)$ and $\pcc(G)=\tilde{\Omega}\left(n^{1.75} \right)$. } Thus, to cryptanalyze a candidate iMHF it would be useful to have an algorithm to test for depth-robustness of an input graph $G$. However, we stress that (constant-factor) hardness of $\REDUCIBLE_d$ does not directly imply that $\minCC$ is $\NPhard$. To the best of our knowledge no one has explored the computational complexity of testing whether a given DAG $G$ is $(e,d)$-depth robust. 

We have many constructions of depth-robust graphs~\cite{EGS75,PR80,Sch82,Sch83,ITCS:MahMorVad13}, but the constant terms in these constructions are typically not well understood. For example,  Erd\"os, Graham and Szemer\'edi~\cite{EGS75} constructed an $\big( \Omega(n),\Omega(n)\big)$-depth robust graph with $\indeg = O\big(\log n\big)$. Alwen \etal~\cite{EC:AlwBloPie17} showed how to transform an $n$ node $(e,d)$-depth robust graph with maximum indegree $\indeg$ to a $(e,d\times \indeg)$-depth robust graph with maximum $\indeg=2$ on $n\times\indeg$ nodes. Applying this transform to the  Erd\"os, Graham and Szemer\'edi~\cite{EGS75} construction yields a constant-indegree graph on $n$ nodes such that $G$ is $(\Omega(n/\log(n)),\Omega(n)\big)$-depth robust --- implying that $\pcc(G) =  \Omega\big(\frac{n^2}{\log n}\big)$. From a theoretical standpoint, this is essentially optimal as any constant $\indeg$ DAG has $\pcc = O\big( \frac{n^2 \log \log n}{\log n}\big)$~\cite{AlwenB16a}. From a practical standpoint it is important to understand the exact values of $e$ and $d$ for specific parameters $n$ in each construction. 

\subsection{Results}
We first produce a reduction from Vertex Cover which preserves approximation hardness. Let $\MINREDUCIBLE_d$ denote the problem of finding a minimum size $S \subseteq V$ such that $\depth(G-S)\leq d$. Our reduction shows that, for each $0 \leq d \leq n^{1-\epsilon}$, it is $\NPhard$ to $1.3$-approximate $\MINREDUCIBLE_d$ since it is $\NPhard$ $1.3$-approximate Vertex Cover~\cite{dinur2005hardness}. Similarly, it is hard to $(2-\epsilon)$-approximate $\MINREDUCIBLE_d$ for any fixed $\epsilon>0$~\cite{khot2008vertex}, under the Unique Games Conjecture~\cite{khot2002power}. We also produce a reduction from Cubic Vertex Cover to show $\REDUCIBLE_d$ is $\NPcomplete$ even when the input graph has bounded indegree. 

The techniques we use are similar to those of Bresar et al.~\cite{BresarKKS11} who considered the problem of finding a minimum size $d$-path cover in undirected graphs (i.e., find a small set $S\subseteq V$ of nodes such that every undirected path in $G-S$ has size at most $d$). However, we stress that if $G$ is a DAG, $\hat{G}$ is the corresponding undirected graph and $S\subseteq V$ is given such that $\depth(G-S) \leq d$ that this does not ensure that $\hat{G}-S$ contains no {\em undirected path} of length $d$. Thus, our reduction specifically address the needs for directed graphs and bounded indegree. 

\newcommand{\reducibleNPC}{$\REDUCIBLE_d$ is $\NPcomplete$ and it is $\NPhard$ to $1.3$-approximate $\MINREDUCIBLE_d$. Under the Unique Games Conjecture, it is hard to $(2-\epsilon)$-approximate $\MINREDUCIBLE_d$.}
\begin{theorem}
\thmlab{reducible:npc}
\reducibleNPC
\end{theorem}

\newcommand{\reducibleIndegNPC}{Even for $\delta=O(1)$, $\REDUCIBLE_{d,\delta}$ is $\NPcomplete$.}
\begin{theorem}
\thmlab{reducible:indeg:npc}
\reducibleIndegNPC
\end{theorem}

We defer the proofs of \thmref{reducible:npc} and \thmref{reducible:indeg:npc} to \appref{MissingProofs}. 
We leave open the question of efficient approximation algorithms for $\MINREDUCIBLE_d$. Lee \cite{Lee17} recently proposed a FPT $O\left(\log~d\right)$-approximation algorithm for the related problem $d$-path cover problem running in time $2^{O(d^3\log d)}n^{O(1)}$. However, it is not clear whether the techniques could be adapted to handle directed graphs and in most interesting cryptographic applications we have $d = \Omega\left(\sqrt{n}\right)$ so the algorithm would not be tractable. 

%% file: approximation.tex
\section{LP Relaxation has Large Integrality Gap} \seclab{integralitygap}
In this section, we show that the natural LP relaxation for the integer program of DAG pebbling has a large integrality gap. 
We show similar results for the natural LP relaxation for the integer program of $\REDUCIBLE_d$ in \appref{reducibleGap}. 

Let $G=(V,E)$ be a DAG with maximum indegree $\delta$ and with $V=\{1,\ldots,n\}$, where $1,2,3,\ldots,n$ is topological ordering of $V$. We start with an integer program for DAG pebbling, in Figure~\ref{fig:IP}.
\begin{figure}[htb]
\vspace{-0.1in}
\begin{tabular}{ll}
{\small $\quad\min \sum_{v \in V} \sum_{t=0}^{n^2} x_v^t$. s.t.} &
{\small (3) $\forall v\in\sinks(G)$, $\sum_{t=0}^{n^2} x_v^t \geq 1$. }\\
{\small (1) $x_v^t \in \{0,1\}$~~ $\forall 1 \leq v \leq n$ and $0\leq t \leq n^2$.}  &
{\small (4) $\forall v \in V \backslash \mathsf{sources}(G)$, $0 \leq t \leq n^2-1$}  \\
{\small (2) $x_v^0 = 0$ ~~~~~~ $\forall 1 \leq v \leq n$.}  &
{\small $\quad\quad x_v^{t+1} \leq x_v^{t} + \frac{\sum_{v' \in\parents{(v)}} x_{v'}^t}{\left|\parents{(v)} \right|}\,.$}
\end{tabular}
\caption{Integer Program for Pebbling.} \label{fig:IP}
\vspace{-0.1in}
\end{figure}
Intuitively, $x_v^t = 1$ if we have a pebble on node $v$ during round $t$. Thus, Constraint 3 says that we must have a pebble on the final node at some point. Constraint 2 says that we do not start with any pebbles on $G$ and Constraint 4 enforces the validity of the pebbling. That is, if $v$ has parents we can only have a pebble on $v$ in round $t+1$ if either (1) $v$ already had a pebble during round $t$, or (2) all of $v$'s parents had pebbles in round $t$. It is clear that the above Integer Program yields the optimal pebbling solution. 

We would like to convert our Integer Program to a Linear Program. The natural relaxation is simply to allow $0 \leq x_v^t \le 1$. However, we show that this LP has a large integrality GAP $\tilde{\Omega}\left(\frac{n}{\log n}\right)$ even for DAGs $G$ with constant indegree. In particular, there exist DAGs with constant indegree $\delta$ for which the optimal pebbling has $\pcc(G) = \tilde{\Omega}\left(\frac{n^2}{\log n}\right)$~\cite{EC:AlwBloPie17}, but we will provide a fractional solution to the LP relaxation with cost $O\big(n\big)$. 


\newcommand{\thmIPgap}{Let $G$ be a DAG. Then there is a fractional solution to our LP Relaxation (of the Integer Program in Figure \ref{fig:IP}) with cost at most $4n$.}
\begin{theorem} \thmlab{IPgap}
\thmIPgap
\end{theorem}
\begin{proof}
In particular, for all time steps $t\le n$, we set $x_i^t=\frac{1}{n}$ for all $i\le t$, and $x_i^t=0$ for $i>t$. 
For time steps $n<t\le n+\ceil{\log n}$, we set $x_v^t=\min\left(1,\frac{2^{t-n}}{n}\right)$ for all $v\in V$. 
We first argue that this is a feasible solution. 
Trivially, Constraints 1 (with the LP relaxation) and 2 are satisfied. 
Moreover, for $t=n+\ceil{\log n}$, $x_v^t=1$ for all $v\in V$, so Constraint 3 is satisfied.
Note that for $1<t\le n$, if $x_v^t=\frac{1}{n}$, then $x_u^{t-1}=\frac{1}{n}$ for all $u<v$, so certainly $\sum_{v' \in\parents(v)} \frac{x_{v'}^{t-1}}{\left|\parents(v) \right|}\ge\frac{1}{n}$. 
Furthermore, $x_v^{t-1}=\min\left(1,\frac{2^{t-1-n}}{n}\right)$ for $n<t\le n+\ceil{\log n}$ implies that setting
\[x_v^t=x_v^{t-1}+\sum_{v' \in\parents(v)} \frac{x_{v'}^t}{\left|\parents(v) \right|}=\frac{2^{t-1-n}}{n}+\frac{2^{t-1-n}}{n}=\frac{2^{t-n}}{n}\]
is valid. 
Therefore, Constraint 4 is satisfied.

Finally, we claim that $\sum_{v \in V} \sum_{t \leq n+\ceil{\log n}} x_{v}^t \leq 4n$. To see this, note that for every round $t\le n$, we have $x_{v}^t \leq \frac{1}{n}$ for all $v\in V$. Thus, \[\sum_{v \in V} \sum_{t \leq n} x_{v}^t \leq \sum_{v \in V} \sum_{t \leq n}\frac{1}{n}\le\sum_{v\in V}1=n.\]
For time steps $n<t<n+\ceil{\log n}$, note that $x_v^t\le\frac{2^{t-n}}{n}$ for all $v\in V$. Thus,
\[\sum_{v \in V} \sum_{n<t<n+\ceil{\log n}} x_{v}^t \leq \sum_{n<t<n+\ceil{\log n}}2^{t-n}\le 2n.\]
Finally, for time step $t=n+\ceil{\log n}$, $x_v^t=1$ for all $v\in V$. Consequently,
\[\sum_{v \in V} \sum_{t=n+\ceil{\log n}} x_{v}^t=n.\] 
Therefore,
\[\sum_{v \in V} \sum_{t \leq n+\ceil{\log n}} x_{v}^t \leq 4n.\]
\end{proof}
%

One tempting way to ``fix'' the linear program is to require that the pebbling take at most $n$ steps since the fractional assignment used to establish the integrality gap takes $2n$ steps. There are two issues with this approach: (1) It is not true in general that the optimal pebbling of $G$ takes at most $n$ steps. See \appref{PebblingTimeExample} for a counter example and discussion. (2) There exists a family of DAGs with constant indegree for which we can give a fractional assignment that takes exactly $n$ steps and costs $O(n \log n)$. Thus, the integrality gap is still $\tilde{\Omega}(n)$. Briefly, in this assignment we set $x_i^i = 1$ for $i \leq n$ and for $i < t$ we set $x^t_i = \max\left\{ \frac{1}{n}, \{2^{-\mathbf{dist}(i,t+j)-j+2 } : j \geq 1 \} \right\}$, where $\mathbf{dist}(x,y)$ is the length of the shortest path from $x$ to $y$. In particular, if $j=1$ (we want to place a `whole' pebble on node $j+t$ in the next round by setting $x^{j+t}_{j+t} = 1$ ) and node $i$ is a parent of node $t+j$ then we have $\mathbf{dist}(i,t+j)=1$ so we will have $x^{j+t}_{i} = 1$ (e.g., a whole pebble on node $i$ during the previous round). 

%% file: future.tex
\vspace{-0.12in}
\section{Conclusions}
\seclab{open}
\vspace{-0.06in}
We initiate the study of the computational complexity of cumulative cost minimizing pebbling in the parallel black pebbling model. 
This problem is motivated by the urgent need to develop and analyze secure data-independent memory hard functions for password hashing. 
We show that it is NP-Hard to find a parallel black pebbling minimizing cumulative cost, and we provide evidence that the problem is hard to approximate. 
Thus, it seems unlikely that we will be able to develop tools to automate the task of a cryptanalyst to obtain strong upper/lower bounds on the security of a candidate iMHF. 
However, we cannot absolutely rule out the possibility that an efficient approximation algorithm exists. 
In fact, our results only establish worst case hardness of graph pebbling. We cannot rule out the existance of efficient algorithms to find optimal pebblings for practical iMHF proposals such as Argon2i~\cite{Argon2} and DRSample~\cite{CCS:AlwBloHar17}. 
The primary remaining challenge is to either give an efficient $\alpha$-approximation algorithm to find a pebbling $P \in \pPeb$ with $\cc(P) \leq \alpha \pcc(G)$ or show that $\pcc(G)$ is hard to approximate. 
We believe that the problem offers many interesting theoretical challenges and a solution could have very practical consequences for secure password hashing. 
It is our hope that this work encourages others in the TCS community to explore these questions.

%% file: appendix.tex
\appendix
\section{Missing Proofs} \applab{MissingProofs}
\input{missingproof}

\section{Integrality Gaps for $\REDUCIBLE_d$} \applab{reducibleGap}
We now suggest a natural integer program for $\REDUCIBLE_d$ and show that the integrality gap is quite large $\tilde{\Omega}(n)$.
\begin{figure}[htb]
\begin{tabular}{ll}
{\small $\quad\min \sum_{v \in V} x_v$. s.t.} &
{\small (2) $0 \leq d_{u,v} \leq d$~~ $\forall (u,v) \in V^2$  } \\
{\small (1) $x_v \in \{0,1\}$~~ $\forall 1 \leq v \leq n$.}  &
{\small (3) $d_{w,v} \geq d_{w,u}+1-(d+1)(x_u+x_v)$~~~ $\forall w \in V, (u,v) \in E$} 
\end{tabular}
\caption{Integer Program for $\REDUCIBLE_d$.} \label{fig:IPReducible}
\end{figure}

Intuitively, setting $x_v=1$ means that we include $v \in S$ and $d_{u,v}$ represents the length of the maximum length directed path from $u$ to $v$ in the graph $G-S$. Constraint 2 requires that the longest path has length at most $d$, and Constraint 3 ensures that $d_{w,v}$ upper bounds the length of the longest path from $w$ to $v$ in $G-S$. If we have a path from $w$ to $u$ of length $k$ in $G-S$ and $u,v \notin S$ then there is a path of length $k+1$ from $w$ to $v$ in $G-S$ --- if $u \in S$ or $v \in S$ then the $-n(x_u+x_v)$ term effectively eliminates the constraint since this particular path does not exist in $G-S$.

The LP relaxation is obtained by changing Constraint 1 to $0 \leq x_v \leq 1$. To see that the LP relaxation has high integrality gap we first observe that there is a family of $\left( \Omega(n),\Omega(n)\right)$-depth robust DAGs $G_n$ with $\indeg(G_n) \leq \log n$.  By \thmref{IPReducible} the LP relaxation for $G_n$ has a solution with cost at most $n/d = \theta(1)$, but the Integer Program must have cost at least $\Omega(n)$. Thus, the integrality gap is $\Omega(n)$. Even if we require $\indeg(G_n) =2$ then we still have a family of $\left( \Omega(n/\log n),\Omega(n)\right)$-depth robust DAGs~\cite{EC:AlwBloPie17} so we get an integrality gap of $\Omega(n/\log n)$.

\begin{theorem} \thmlab{IPReducible}
For {\em any} DAG $G$ the LP relaxation (of the Integer Program in \figref{IPReducible}) has a solution with cost at most $n/d$.
\end{theorem}
\begin{proof}
Set $x_v = \frac{1}{d}$ for all $v \in V$ and set $d_{u,v} = 0$ for all $u,v \in V^2$. 
\end{proof}

\section{On Pebbling Time and Cumulative Cost} \applab{PebblingTimeExample}
\input{counterexample}

\section{NP-Hardness of \texttt{minST}} \applab{SpaceTime}
Recall that the space-time complexity of a graph pebbling is defined as $\ST(P) = t \times \max_{1 \le i \le t} \left| P_i\right|$. 
We define $\minST$ and $\minSST$ based on whether the graph pebbling is parallel or sequential. 
Formally, the decision problem $\minST$ is defined as follows: \\
 {\noindent \bf Input:} a DAG $G$ on $n$ nodes and an integer $k < n(n+1)/2$. \newline
 {\noindent \bf Output:} {\em Yes}, if $\min_{P\in\pPeb_G}\ST(P)\leq k$; otherwise {\em No}. \newline

\noindent
The decision problem $\minSST$ is defined as follows: \\
 {\noindent \bf Input:} a DAG $G$ on $n$ nodes and an integer $k < n(n+1)/2$. \newline
 {\noindent \bf Output:} {\em Yes}, if $\sst(G) \leq k$; otherwise {\em No}. \newline

Gilbert \etal \cite{gilbert1980pebbling} provide a construction from any instance of TQBF to a DAG $G_{TQBF}$ with pebbling number $3n+3$ if and only if the instance is satisfiable. Here, the pebbling number of a DAG $G$ is $\min_{P = (P_1,\ldots,P_t) \in \pPeb} \max_{i \leq t} \left| P_i \right|$ is the number of pebbles necessary to pebble $G$. An important gadget in their construction is the so-called pyramid DAG. We use a triangle with the number $k$ inside to denote a $k$-pyramid (see \figref{pyramid} for an example of a $3$-pyramid). The key property of these DAGs is that {\em any} legal pebbling $P=(P_0,\ldots,P_t)\in \pPeb(Pyramid_k)$ of a $k$-pyramid  requires at least $\min_{i} \left| P_i\right| \geq k$ pebbles on the DAG at some point in time.
Another gadget, which appears in Figure~\ref{fig:exists}, is the existential quantifier gadget, which requires that $s_i$, $s_i-1$, and $s_i-2$ pebbles must be placed in each of the pyramids to ultimately pebble $q_i$.   \\

\noindent {\bf Remark:} We note that \cite{gilbert1980pebbling} focused on sequential pebblings $(P \in \Peb)$ in their analysis, but their analysis extends to parallel pebblings $(P \in \pPeb)$ as well.

\begin{figure}[ht]
\centering
\begin{tikzpicture}[scale=1.4]
\draw[black,fill=black]  (4.5cm,2cm) circle (0.1cm); 
\draw (4cm,1cm) circle (0.1cm);   
\draw (5cm,1cm) circle (0.1cm); 
\draw (3.5cm,0cm) circle (0.1cm); 
\draw (4.5cm,0cm) circle (0.1cm); 
\draw (5.5cm,0cm) circle (0.1cm);   

\draw[->, black] (5.5cm,0.1cm) -- (5.1cm,0.9cm);
\draw[->, black] (4.5cm,0.1cm) -- (4.9cm,0.9cm);
\draw[->, black] (4.5cm,0.1cm) -- (4.1cm,0.9cm);
\draw[->, black] (3.5cm,0.1cm) -- (3.9cm,0.9cm);

\draw[->, black] (4cm,1.1cm) -- (4.4cm,1.9cm);
\draw[->, black] (5cm,1.1cm) -- (4.6cm,1.9cm);

\node at (2cm+5cm,-1.7cm+3cm){$\equiv$};

\draw[black,fill=black] (10cm,1.5cm) circle (0.1cm);
\draw[black] (9.95cm,-1.0866cm+2.5cm) -- (9.5cm,-1.866cm+2.5cm) -- (2cm+8.5cm,-1.866cm+2.5cm) -- (2cm+8.05cm,-1.0866cm+2.5cm);
\node at (2cm+8cm,-1.7cm+2.5cm){$k$};
\end{tikzpicture}
\caption{A $3$-Pyramid.}\label{fig:pyramid}
\end{figure}
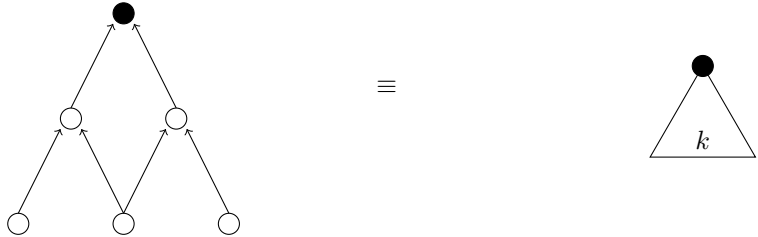

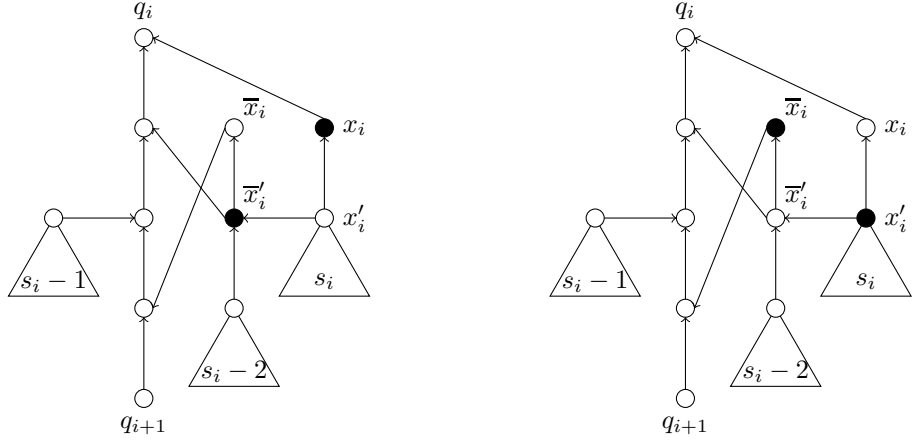
\begin{figure}[ht]
\centering
\begin{tikzpicture}[scale=1.2]
\draw (0cm,0cm) circle (0.1cm);
\draw[black] (-0.05cm,-0.0866cm) -- (-0.5cm,-0.866cm) -- (0.5cm,-0.866cm) -- (0.05cm,-0.0866cm);
\node at (0cm,-0.7cm){$s_i-1$};

\draw (1cm,0cm) circle (0.1cm);
\draw (1cm,-1cm) circle (0.1cm);
\draw (1cm,-2cm) circle (0.1cm);
\draw (1cm,1cm) circle (0.1cm);
\draw (1cm,2cm) circle (0.1cm);
\draw[->, black] (0.1cm,0cm) -- (0.9cm,0.cm);
\draw[->, black] (1cm,-1.9cm) -- (1cm,-1.1cm);
\draw[->, black] (1cm,-0.9cm) -- (1cm,-0.1cm);
\draw[->, black] (1cm,0.1cm) -- (1cm,0.9cm);
\draw[->, black] (1cm,1.1cm) -- (1cm,1.9cm);

\draw (2cm,-1cm) circle (0.1cm);
\draw[black] (1.95cm,-1.0866cm) -- (1.5cm,-1.866cm) -- (2.5cm,-1.866cm) -- (2.05cm,-1.0866cm);
\node at (2cm,-1.7cm){$s_i-2$};

\draw[black,fill=black] (2cm,0cm) circle (0.1cm);
\draw (2cm,1cm) circle (0.1cm);
\draw[->, black] (2cm,-0.9cm) -- (2cm,-0.1cm);
\draw[->, black] (2cm,0.1cm) -- (2cm,0.9cm);

\draw (3cm,0cm) circle (0.1cm);
\draw[black] (2.95cm,-0.0866cm) -- (2.5cm,-0.866cm) -- (3.5cm,-0.866cm) -- (3.05cm,-0.0866cm);
\node at (3cm,-0.7cm){$s_i$};

\draw[black,fill=black] (3cm,1cm) circle (0.1cm);
\draw[->, black] (3cm,0.1cm) -- (3cm,0.9cm);

\draw[->, black] (2.9cm,0cm) -- (2.1cm,0cm);
\draw[->, black] (1.9cm,1cm) -- (1.1cm,-1cm);
\draw[->, black] (1.9cm,0cm) -- (1.1cm,1cm);
\draw[->, black] (3cm,1.1cm) -- (1.1cm,2cm);

\node[below] at (1cm,-2.1cm){$q_{i+1}$};
\node[above] at (1cm,2.1cm){$q_i$};
\node[right] at (3.1cm,0cm){$x'_i$};
\node[right] at (3.1cm,1cm){$x_i$};
\node[above right] at (2cm,0cm){$\overline{x}'_i$};
\node[above right] at (2cm,1cm){$\overline{x}_i$};

\draw (6cm,0cm) circle (0.1cm);
\draw[black] (5.95cm,-0.0866cm) -- (5.5cm,-0.866cm) -- (6.5cm,-0.866cm) -- (6.05cm,-0.0866cm);
\node at (6cm,-0.7cm){$s_i-1$};

\draw (7cm,0cm) circle (0.1cm);
\draw (7cm,-1cm) circle (0.1cm);
\draw (7cm,-2cm) circle (0.1cm);
\draw (7cm,1cm) circle (0.1cm);
\draw (7cm,2cm) circle (0.1cm);
\draw[->, black] (6.1cm,0cm) -- (6.9cm,0.cm);
\draw[->, black] (7cm,-1.9cm) -- (7cm,-1.1cm);
\draw[->, black] (7cm,-0.9cm) -- (7cm,-0.1cm);
\draw[->, black] (7cm,0.1cm) -- (7cm,0.9cm);
\draw[->, black] (7cm,1.1cm) -- (7cm,1.9cm);

\draw (8cm,-1cm) circle (0.1cm);
\draw[black] (7.95cm,-1.0866cm) -- (7.5cm,-1.866cm) -- (8.5cm,-1.866cm) -- (8.05cm,-1.0866cm);
\node at (8cm,-1.7cm){$s_i-2$};

\draw (8cm,0cm) circle (0.1cm);
\draw[black,fill=black] (8cm,1cm) circle (0.1cm);
\draw[->, black] (8cm,-0.9cm) -- (8cm,-0.1cm);
\draw[->, black] (8cm,0.1cm) -- (8cm,0.9cm);

\draw[black,fill=black] (9cm,0cm) circle (0.1cm);
\draw[black] (8.95cm,-0.0866cm) -- (8.5cm,-0.866cm) -- (9.5cm,-0.866cm) -- (9.05cm,-0.0866cm);
\node at (9cm,-0.7cm){$s_i$};

\draw (9cm,1cm) circle (0.1cm);
\draw[->, black] (9cm,0.1cm) -- (9cm,0.9cm);

\draw[->, black] (8.9cm,0cm) -- (8.1cm,0cm);
\draw[->, black] (7.9cm,1cm) -- (7.1cm,-1cm);
\draw[->, black] (7.9cm,0cm) -- (7.1cm,1cm);
\draw[->, black] (9cm,1.1cm) -- (7.1cm,2cm);

\node[below] at (7cm,-2.1cm){$q_{i+1}$};
\node[above] at (7cm,2.1cm){$q_i$};
\node[right] at (9.1cm,0cm){$x'_i$};
\node[right] at (9.1cm,1cm){$x_i$};
\node[above right] at (8cm,0cm){$\overline{x}'_i$};
\node[above right] at (8cm,1cm){$\overline{x}_i$};

\end{tikzpicture}
\caption{An existential quantifier, with $x_i$ set to true in the left figure and $x_i$ set to false in the right figure.}\label{fig:exists}
\end{figure}

We observe that any instance of TQBF where each quantifier is an existential quantifier requires at most a quadratic number of pebbling moves. 
Specifically, we look at instances of \texttt{3-SAT}, such as in Figure~\ref{fig:graph}. 
In such a graph representing an instance of \texttt{3-SAT}, the sink node to be pebbled is $q_n$. 
By design of the construction, any true statement requires exactly three pebbles for each pyramid representing a clause. 
On the other hand, a false clause requires four pebbles, so that false statements require more pebbles. 
Thus, by providing extraneous additions to the construction which force the number of pebbling moves to be a known constant, we can extract the pebbling number, given the space-time complexity. 
For more details, see the full description in \cite{gilbert1980pebbling}.

\begin{lemma}\cite{gilbert1980pebbling}
\lemmlab{peb:number}
The quantified Boolean formula $Q_1x_1Q_2x_2\cdots Q_nx_nF_n$ is true if and only if the corresponding DAG $G_{TQBF}$ has pebbling number $3n+3$. 
\end{lemma}

\begin{lemma}
\lemmlab{TQBF:moves}
Suppose that we have a satisfiable TQBF formula $Q_1x_1Q_2x_2\cdots Q_nx_nF_n$ with $Q_i = \exists$ for all $i \leq n$. Then there is a legal sequential pebbling $P=(P_0,\ldots,P_t) \in \Peb\big(G_{TQBF}\big)$ of the corresponding DAG $G_{TQBF}$ from \cite{gilbert1980pebbling} with $t \leq 6n^2+33n$ pebbling moves and $\max_{i \leq t} \left| P_i \right| \leq 3n+3$.
\end{lemma}
\begin{proof}
We describe the pebbling strategy of Gilbert \etal~\cite{gilbert1980pebbling}, and analyze the pebbling time of their strategy. Let $T(i)$ be the time it takes to pebble $q_i$ in the proposed construction for any instance with $i$ variables, $i$ clauses, and only existential quantifiers.

Suppose that $x_i$ is allowed to be true for the existential quantifier $Q_i=\exists$. 
Then vertex $x'_i$ is pebbled using $s_i$ moves, where $s_i=3n-3i+6$. 
Similarly, vertices $d_i$ and $f_i$ are pebbled using $s_i-1$ and $s_i-2$ moves respectively. 
Additionally, $f_i$ is moved to $\overline{x}'_i$ and then $x_i$ is moved to $x'_i$ in the following step, for a total of two more moves. 
We then pebble $q_{i+1}$ using $T(i+1)$ moves and finish by placing a pebble on $\overline{x}_i$ and moving it to $c_i$, $b_i$, $a_i$, and $q_i$, for five more moves. 
Finally, we use six more moves to pebble an additional clause.
Thus, in this case,
\[T_{true}(i)=s_i + (s_i-1) + (s_i-2) + 13 + T(i+1).\]

On the other hand, if $x_i$ is allowed to be false for the existential quantifier $Q_i=\exists$, then first we pebble $x'_i$, $d_i$, and $f_i$ sequentially, using $s_i$, $s_i-1$, and $s_i-2$ moves respectively. 
We then move the pebble from $f_i$ to $\overline{x}'_i$ and then to $\overline{x}_i$, for a total of two more moves.
We then pebble $q_{i+1}$ using $T(i+1)$ moves. 
The pebble on $q_{i+1}$ is subsequently moved to $c_i$ and then $b_i$, using two more moves. 
Picking up all pebbles except those on $b_i$ and $x'_i$, and using them to pebble $f_i$ takes $s_i-2$ more moves. 
Additionally, the pebble on $f_i$ is moved to $\overline{x}'_i$ and then $a_i$, while the pebble on $x'_i$ is moved to $x_i$ and then $q_i$, for four more moves. 
Finally, we use six more moves to pebble an additional clause.
In total,
\[T_{false}(i) = s_i + (s_i-1) + (s_i-2) + (s_i-2) + 14 + T(i+1).\]
\noindent
Therefore,
\[T(i)\le 4s_i + 10 + T(i+1),\]
where $s_i=3n-3i+6$. Thus,
\[T(i)\le 12(n-i) + 34 + T(i+1).\]
Writing $R(i)=T(n-i)$ then gives
\[R(i)\le 12i + 34 + R(i-1),\]
so that $R(n)\le\sum_{i=1}^n(12i+34)=6n^2+40n$. Hence, it takes at most $6n^2+40n$ moves to pebble the given construction for any instance of TQBF which only includes existential quantifiers.
\end{proof}

\begin{theorem} \thmlab{minSTNPComplete}
$\minST$ is $\NPcomplete$.
\end{theorem}
\begin{proof}
We provide a reduction from \texttt{3-SAT} to \texttt{minST}. Given an instance $\mathcal{I}$ of \texttt{3-SAT} with at most $n$ clauses or variables, we create the corresponding graph from \cite{gilbert1980pebbling}. Additionally, we append a chain of length $300n^3+6n^2+40n+100$ to the graph with an edge from the sink node $q_n$ from \cite{gilbert1980pebbling} to the first node in our chain. By adding a chain of length $300n^3+6n^2+40n+100$ we can ensure that for any {\em legal} pebbling $P=(P_0,\ldots,P_t) \in \pPeb(G)$ of our graph $G$ we have $t \geq 300n^3+6n^2+40n+100$.  By \lemmref{TQBF:moves}, at most $6n^2+40n$ moves are necessary to pebble the $3-SAT$ portion of the graph, while the chain requires exactly $300n^3+6n^2+40n+100$ moves. Thus, if $\mathcal{I}$ is satisfiable then we can find a legal pebbling $P=(P_0,\ldots,P_t)$ with space $\max_{i \leq t} |P_i| \leq 3n+3$ and $t \leq 300n^3+12n^2+80n+100$ moves. First, pebble the sink $q_n$ in $t' = 6n^2+40n$ steps and $\max_{i \leq t'} |P_i| \leq 3n+3$ space by \lemmref{TQBF:moves} and then walk single pebble down the chain in $300n^3+6n^2+40n+10$ steps keeping at most one pebble on the DAG in each step. The space time cost is at most $\ST(P) \leq 900n^4+936n^3+276n^2+540n+300$. If  $\mathcal{I}$ is not satisfiable then, by \lemmref{peb:number} for {\em any legal} pebbling $P=(P_0,\ldots,P_t) \in \pPeb(G)$ of our graph we have $\max_{i \leq t} |P_i| \geq 3n+4$ and $t \geq 300n^3+6n^2+40n+100$. Thus, $\ST(P) \geq 900n^4+1218 n^3+144 n^2 +460n + 400$ we have 
\[900n^4+1218 n^3+144 n^2 +460n + 400 - \big(900n^4+936n^3+276 n^2+540n+300 \big) > 0 \ . \]
for all $n>0$.  Thus,  $\mathcal{I}$ is satisfiable if and only if there exists a legal pebbling with $\ST(P) \leq  900n^4+936n^3+276 n^2+540n+300$. Clearly, this reduction can be done in polynomial time, and so there is indeed a polynomial time reduction from \texttt{3-SAT} to \texttt{minST}.
\end{proof}

\noindent
We note that the same reduction from \textsf{TQBF} to \texttt{minST} fails, since there exist instances of TQBF where the pebbling time is  $2^{\Omega(n)}$.
However, the same relationships do hold for sequential pebbling.

The proof of \thmref{minSSTNPComplete} is the same as the proof of \thmref{minSTNPComplete} because we can exploit the fact that the pebbling of $G_{TQBF}$ from \lemmref{TQBF:moves} is sequential. 
\begin{theorem} \thmlab{minSSTNPComplete}
$\minSST$ is $\NPcomplete$.
\end{theorem}

\section{Discussion of Other Techniques} \seclab{OtherTechniques}
In this section, we address a number of seemingly similar problems. 
In particular, several related graph partitioning problems actually have fundamentally different structures.

In the $d$-Vertex Separator problem, the goal is to remove the minimum number of vertices so that each remaining connected component has at most $d$ vertices. 
This problem fails to translate to a successful solution for $(e,d)$-reducibility, as a binary tree of height $d-1$ with all directed edges pointing from parents to leaves has an exponential number of vertices, but requires no removal of vertices to ensure depth less than $d$. 

In the $d$-Distance Minimum Vertex Cover problem, the goal is to find the smallest subset $S$ of vertices so that all vertices are at most distance $d$ from a vertex in $S$. 
Of course, even if all vertices are within distance $1$ from a vertex in $S$, the depth of $G-S$ can be as large as $\frac{n}{2}$. 
Consider, for example, a path of length $n$, with additional edges $(i,i+2)$ for each $i\le n-2$. 
Then even removing all even or all odd labeled vertices leaves a path of length $\frac{n}{2}$.

Recently, Lee \cite{Lee17} proposed an $O\left(\log~d\right)$-approximation algorithm to the $d$-Path Transversal problem (also called $d$-path cover by Bresar et al.~\cite{BresarKKS11}) for undirected graphs with parameterized complexity $2^{O(d^3\log d)}n^{O(1)}$. The goal is to remove the minimum number of vertices so that the resulting graph contains no (undirected) paths of length $d$. It is not clear whether or not the techniques could be extended to deal with directed graphs. Furthermore, in most cryptanalysis applications we will have $d=\Omega(\sqrt[3]{n})$ so the approximation algorithm would run in exponential time. 

\ignore{\paragraph{Reduction.} We note that $d$-Path cover/traversal problem can be reduced to $\MINREDUCIBLE_d$. Thus, $\MINREDUCIBLE_d$ is at least as hard to approximate as the $d$-path cover problem. The basic idea is to create $d$ copies $v^1,\ldots, v^d$ of each vertex $v \in V$. If the original (undirected) graph $G$ has an undirected edge $\{u,v\}$ then we add each of the directed edges $(u^i,v^{i+1})$ and $(v^i,u^{i+1})$ for each $i < d$ to the DAG $D$. It is easy to see that any $d$-path cover $S \subseteq V$ in $G$ corresponds to a depth reducing set $S'$ of the same size such that $\depth(D-S') \leq d$ (if $v \in S$ the add $v^1$ to $S'$). Similarly, any depth reducing set $S'$ for $D$ can be mapped to a $d$-path cover $|S| \leq |S'|$.}

Another possible approach to build an approximation algorithm for $\MINREDUCIBLE_d$ would be to exploit tree embeddings by transforming the DAG $G$ into a tree using the longest path metric and then find a depth-reducing set for the resulting tree. However, the longest path between two vertices in directed graphs does not qualify as a metric, and it is not immediately obvious how to address this issue. Furthermore, even if one could find a tree embedding for DAGs which approximately preserves distances under the longest path metric it is not clear how to use a depth-reducing set for the tree to produce a depth-reducing set in the original DAG. In particular, observe that a complete DAG and a path of length $n$ might both yield a simple path after performing the tree embedding under the longest path metric. However, the size of the depth-reducing sets of the two instances differ drastically. 

\subsection*{Why Doesn't the Gilbert \etal~Reduction Work for Cumulative Cost?}

The construction from \cite{gilbert1980pebbling} is designed to minimize the number of pebbles simultaneously on the graph, at the expense of larger number of necessary time steps and/or a larger average number of pebbles on the graph during each pebbling round. As a result, one can bypass several time steps by simply adding additional pebbles during some time step. For example, it may be beneficial to temporarily keep pebbles on all three nodes $x_i$, $\overline{x}_i$ and $\overline{x}_i'$ at times so that we can avoid repebbling the $s_{i}-2$-pyramid later. Also if we do not need the pebble on node $x_i$ for the next {$s_i \choose 2$} steps then it is better to discard any pebbles on $x_i'$ and $x_i$ entirely to reduce cumulative cost. Because we would need to repebble the $s_i$-pyramid later our maximum space usage will increase, but our cumulative cost would decrease.  Our reduction to space-time cost works because $\ST$ cost is highly sensitive to an increase in the number of pebbles on the graph even if this increase is temporary. Cumulative, unlike space cost or space-time cost, is not very sensitive to such temporary increases in the number of pebbles on the graph. 

\begin{figure}[ht]
\centering
\begin{tikzpicture}[scale=1.2]
\draw (0cm,0cm) circle (0.1cm);
\draw[black] (-0.05cm,-0.0866cm) -- (-0.5cm,-0.866cm) -- (0.5cm,-0.866cm) -- (0.05cm,-0.0866cm);
\node at (0cm,-0.7cm){$5$};

\draw (1cm,0cm) circle (0.1cm);
\draw (1cm,-1cm) circle (0.1cm);
\draw (1cm,1cm) circle (0.1cm);
\draw (1cm,2cm) circle (0.1cm);
\draw[->, black] (0.1cm,0cm) -- (0.9cm,0cm);
\draw[->, black] (1cm,-0.9cm) -- (1cm,-0.1cm);
\draw[->, black] (1cm,0.1cm) -- (1cm,0.9cm);
\draw[->, black] (1cm,1.1cm) -- (1cm,1.9cm);

\draw (2cm,-1cm) circle (0.1cm);
\draw[black] (1.95cm,-1.0866cm) -- (1.5cm,-1.866cm) -- (2.5cm,-1.866cm) -- (2.05cm,-1.0866cm);
\node at (2cm,-1.7cm){$4$};

\draw (2cm,0cm) circle (0.1cm);
\draw (2cm,1cm) circle (0.1cm);
\draw[->, black] (2cm,-0.9cm) -- (2cm,-0.1cm);
\draw[->, black] (2cm,0.1cm) -- (2cm,0.9cm);

\draw (3cm,0cm) circle (0.1cm);
\draw[black] (2.95cm,-0.0866cm) -- (2.5cm,-0.866cm) -- (3.5cm,-0.866cm) -- (3.05cm,-0.0866cm);
\node at (3cm,-0.7cm){$6$};

\draw (3cm,1cm) circle (0.1cm);
\draw[->, black] (3cm,0.1cm) -- (3cm,0.9cm);

\draw[->, black] (2.9cm,0cm) -- (2.1cm,0cm);
\draw[->, black] (1.9cm,1cm) -- (1.1cm,-1cm);
\draw[->, black] (1.9cm,0cm) -- (1.1cm,1cm);
\draw[->, black] (3cm,1.1cm) -- (1.1cm,2cm);

\node[left] at (0.9cm,2cm){$q_4$};
\node[right] at (3.1cm,1cm){$x_4$};
\node[above right] at (2cm,1cm){$\overline{x}_4$};

\draw (0cm,4cm) circle (0.1cm);
\draw[black] (-0.05cm,3.9134cm) -- (-0.5cm,3.134cm) -- (0.5cm,3.134cm) -- (0.05cm,3.9134cm);
\node at (0cm,3.3cm){$8$};

\draw (1cm,4cm) circle (0.1cm);
\draw (1cm,3cm) circle (0.1cm);
\draw (1cm,2cm) circle (0.1cm);
\draw (1cm,5cm) circle (0.1cm);
\draw (1cm,6cm) circle (0.1cm);
\draw[->, black] (0.1cm,4cm) -- (0.9cm,4cm);
\draw[->, black] (1cm,2.1cm) -- (1cm,2.9cm);
\draw[->, black] (1cm,3.1cm) -- (1cm,3.9cm);
\draw[->, black] (1cm,4.1cm) -- (1cm,4.9cm);
\draw[->, black] (1cm,5.1cm) -- (1cm,5.9cm);

\draw (2cm,3cm) circle (0.1cm);
\draw[black] (1.95cm,2.9134cm) -- (1.5cm,2.134cm) -- (2.5cm,2.134cm) -- (2.05cm,2.9134cm);
\node at (2cm,2.3cm){$7$};

\draw (2cm,4cm) circle (0.1cm);
\draw (2cm,5cm) circle (0.1cm);
\draw[->, black] (2cm,3.1cm) -- (2cm,3.9cm);
\draw[->, black] (2cm,4.1cm) -- (2cm,4.9cm);

\draw (3cm,4cm) circle (0.1cm);
\draw[black] (2.95cm,3.9134cm) -- (2.5cm,3.134cm) -- (3.5cm,3.134cm) -- (3.05cm,3.9134cm);
\node at (3cm,3.3cm){$9$};

\draw (3cm,5cm) circle (0.1cm);
\draw[->, black] (3cm,4.1cm) -- (3cm,4.9cm);

\draw[->, black] (2.9cm,4cm) -- (2.1cm,4cm);
\draw[->, black] (1.9cm,5cm) -- (1.1cm,3cm);
\draw[->, black] (1.9cm,4cm) -- (1.1cm,5cm);
\draw[->, black] (3cm,5.1cm) -- (1.1cm,6cm);

\node[left] at (0.9cm,6cm){$q_3$};
\node[right] at (3.1cm,5cm){$x_3$};
\node[above right] at (2cm,5cm){$\overline{x}_3$};

\draw (0cm,8cm) circle (0.1cm);
\draw[black] (-0.05cm,7.9134cm) -- (-0.5cm,7.134cm) -- (0.5cm,7.134cm) -- (0.05cm,7.9134cm);
\node at (0cm,7.3cm){$11$};

\draw (1cm,8cm) circle (0.1cm);
\draw (1cm,7cm) circle (0.1cm);
\draw (1cm,6cm) circle (0.1cm);
\draw (1cm,9cm) circle (0.1cm);
\draw (1cm,10cm) circle (0.1cm);
\draw[->, black] (0.1cm,8cm) -- (0.9cm,8cm);
\draw[->, black] (1cm,6.1cm) -- (1cm,6.9cm);
\draw[->, black] (1cm,7.1cm) -- (1cm,7.9cm);
\draw[->, black] (1cm,8.1cm) -- (1cm,8.9cm);
\draw[->, black] (1cm,9.1cm) -- (1cm,9.9cm);

\draw (2cm,7cm) circle (0.1cm);
\draw[black] (1.95cm,6.9134cm) -- (1.5cm,6.134cm) -- (2.5cm,6.134cm) -- (2.05cm,6.9134cm);
\node at (2cm,6.3cm){$10$};

\draw (2cm,8cm) circle (0.1cm);
\draw (2cm,9cm) circle (0.1cm);
\draw[->, black] (2cm,7.1cm) -- (2cm,7.9cm);
\draw[->, black] (2cm,8.1cm) -- (2cm,8.9cm);

\draw (3cm,8cm) circle (0.1cm);
\draw[black] (2.95cm,7.9134cm) -- (2.5cm,7.134cm) -- (3.5cm,7.134cm) -- (3.05cm,7.9134cm);
\node at (3cm,7.3cm){$12$};

\draw (3cm,9cm) circle (0.1cm);
\draw[->, black] (3cm,8.1cm) -- (3cm,8.9cm);

\draw[->, black] (2.9cm,8cm) -- (2.1cm,8cm);
\draw[->, black] (1.9cm,9cm) -- (1.1cm,7cm);
\draw[->, black] (1.9cm,8cm) -- (1.1cm,9cm);
\draw[->, black] (3cm,9.1cm) -- (1.1cm,10cm);

\node[left] at (0.9cm,10cm){$q_2$};
\node[right] at (3.1cm,9cm){$x_2$};
\node[above right] at (2cm,9cm){$\overline{x}_2$};

\draw (0cm,12cm) circle (0.1cm);
\draw[black] (-0.05cm,11.9134cm) -- (-0.5cm,11.134cm) -- (0.5cm,11.134cm) -- (0.05cm,11.9134cm);
\node at (0cm,11.3cm){$14$};

\draw (1cm,12cm) circle (0.1cm);
\draw (1cm,11cm) circle (0.1cm);
\draw (1cm,10cm) circle (0.1cm);
\draw (1cm,13cm) circle (0.1cm);

\filldraw[shading=radial,inner color=white, outer color=gray!75, opacity=0.2] (1cm,14cm) circle (0.1cm);
\draw[->, black] (0.1cm,12cm) -- (0.9cm,12cm);
\draw[->, black] (1cm,10.1cm) -- (1cm,10.9cm);
\draw[->, black] (1cm,11.1cm) -- (1cm,11.9cm);
\draw[->, black] (1cm,12.1cm) -- (1cm,12.9cm);
\draw[->, black] (1cm,13.1cm) -- (1cm,13.9cm);

\draw (2cm,11cm) circle (0.1cm);
\draw[black] (1.95cm,10.9134cm) -- (1.5cm,10.134cm) -- (2.5cm,10.134cm) -- (2.05cm,10.9134cm);
\node at (2cm,10.3cm){$13$};

\draw (2cm,12cm) circle (0.1cm);
\draw (2cm,13cm) circle (0.1cm);
\draw[->, black] (2cm,11.1cm) -- (2cm,11.9cm);
\draw[->, black] (2cm,12.1cm) -- (2cm,12.9cm);

\draw (3cm,12cm) circle (0.1cm);
\draw[black] (2.95cm,11.9134cm) -- (2.5cm,11.134cm) -- (3.5cm,11.134cm) -- (3.05cm,11.9134cm);
\node at (3cm,11.3cm){$15$};

\draw (3cm,13cm) circle (0.1cm);
\draw[->, black] (3cm,12.1cm) -- (3cm,12.9cm);

\draw[->, black] (2.9cm,12cm) -- (2.1cm,12cm);
\draw[->, black] (1.9cm,13cm) -- (1.1cm,11cm);
\draw[->, black] (1.9cm,12cm) -- (1.1cm,13cm);
\draw[->, black] (3cm,13.1cm) -- (1.1cm,14cm);

\node[left] at (0.8cm,14cm){$q_1$: Sink};
\node[right] at (3.1cm,13cm){$x_1$};
\node[above right] at (2cm,13cm){$\overline{x}_1$};

\draw (6cm,10cm) circle (0.1cm);
\node[right] at (6.15cm,10cm){$p_0$};
\draw[->, black] (6cm,9.9cm) -- (6cm,9.15cm);
\draw (6cm,9cm) circle (0.1cm);
\draw (6cm,8cm) circle (0.1cm);
\draw (6cm,7cm) circle (0.1cm);
\draw (7cm,8.5cm) circle (0.1cm);
\draw (7cm,7.5cm) circle (0.1cm);
\draw (8cm,8cm) circle (0.1cm);
\node[right] at (8.15cm,8cm){$p_1$};
\draw[->, black] (6.1cm,9cm) -- (6.85cm,8.6cm);
\draw[->, black] (6.1cm,8cm) -- (6.85cm,8.4cm);
\draw[->, black] (6.1cm,8cm) -- (6.85cm,7.6cm);
\draw[->, black] (6.1cm,7cm) -- (6.85cm,7.4cm);
\draw[->, black] (7.1cm,8.5cm) -- (7.85cm,8.1cm);
\draw[->, black] (7.1cm,7.5cm) -- (7.85cm,7.9cm);

\draw[->, black] (3.1cm,13cm) -- (5.85cm,9cm);
\draw[->, black] (3.1cm,13cm) -- (5.85cm,8.1cm);
\draw[->, black] (3.1cm,5cm) -- (5.85cm,7.9cm);
\draw[->, black] (3.1cm,5cm) -- (5.85cm,6.9cm);
\draw[->, black] (3.1cm,1cm) -- (6cm,6.85cm);

\draw[->, black] (8cm,7.9cm) -- (7cm,4.15cm);
\draw (7cm,4cm) circle (0.1cm);
\draw (7cm,3cm) circle (0.1cm);
\draw (7cm,2cm) circle (0.1cm);
\draw (8cm,3.5cm) circle (0.1cm);
\draw (8cm,2.5cm) circle (0.1cm);
\draw (9cm,3cm) circle (0.1cm);
\node[right] at (9.15cm,3cm){$p_2=q_5$};
\draw[->, black] (7.1cm,4cm) -- (7.85cm,3.6cm);
\draw[->, black] (7.1cm,3cm) -- (7.85cm,3.4cm);
\draw[->, black] (7.1cm,3cm) -- (7.85cm,2.6cm);
\draw[->, black] (7.1cm,2cm) -- (7.85cm,2.4cm);
\draw[->, black] (8.1cm,3.5cm) -- (8.85cm,3.1cm);
\draw[->, black] (8.1cm,2.5cm) -- (8.85cm,2.9cm);

\draw[->, black] (3.1cm,9cm) -- (6.85cm,4cm);
\draw[->, black] (3.1cm,9cm) -- (6.85cm,3.1cm);
\draw[->, black] (3.1cm,5cm) -- (6.85cm,2.9cm);
\draw[->, black] (3.1cm,5cm) -- (6.85cm,2.15cm);
\draw[->, black] (2.1cm,1cm) -- (6.85cm,1.85cm);

\draw[->, black] (9cm,2.85cm) -- (9cm,-2.5cm) -- (1cm,-2.5cm) -- (1cm,-1.15cm);
\end{tikzpicture}
\caption{Graph $G_{TQBF}$ for $\exists x_1, x_2,x_3,x_4$ s.t. $(x_1\lor x_2\lor x_4)\land(x_2\lor x_3\lor\overline{x}_4)$.}\label{fig:graph}
\end{figure}

%% file: missingproof.tex

\begin{reminderlemma}{\lemmref{exists:solution}}
\LemmaExistSolution
\end{reminderlemma}
\newline
\begin{proofof}{\lemmref{exists:solution}}
Suppose the given instance of \BTwoLC has a valid solution, $\{x_{i,j}\}$. 
Recall that a pebble must pass $m$ times through each of the $\tau$ chains of length $c$ representing each variable. 
For a set $k$, let $x_{k,j_1}\le x_{k,j_2}\le\ldots\le x_{k,j_n}$. 
We start the $k$\th pass through the $\tau$ chains by placing a pebble on $C_{j_n}^1,\ldots,C_{j_n}^{\tau}$, the $\tau$ chains representing variable $j_n$. 
At each subsequent time step, we move the existing pebbles to the next node along the chain. 
When the pebbles on chains $C_{j_n}^1,\ldots,C_{j_n}^\tau$ reach the $(x_{k,j_n}-x_{k,j_{n-1}}+1)$\th nodes, we place a pebble on $C_{j_{n-1}}^1,\ldots,C_{j_{n-1}}^\tau$, the $\tau$ chains representing variable $j_{n-1}$. 
We continue this process by moving existing pebbles to the next node along each chain for each subsequent time step. 

When pebbles on chains $C_{j_l}^1,\ldots,C_{j_l}^\tau$ reach the $(x_{k,j_l}-x_{k,j_{l-1}}+1)$\th nodes, we place a pebble on $C_{j_{l-1}}^1,\ldots,C_{j_{l-1}}^\tau$, the $\tau$ chains representing variable $j_{l-1}$. 
Thus, the gadgets $E_1,\ldots,E_m$ which are satisfied by $x_{k,1},\ldots,x_{k,m+1}$ can be pebbled during this pass, since by construction, we offset the positions of the pebbles by the appropriate distances. 
When a pebble reaches the end of its chain, we remove the pebble. 
Hence, we see that each chain has $c$ time steps with pebbles, and each of those steps needs only one pebble. 
There are $n$ variables, $\tau$ chains representing each variable, and $m$ passes through each chain.
Across all variable chains, the cumulative complexity is $\tau cmn$ since there are $n$ variables and $\tau$ chains representing each variable. 

Similarly, making $m$ walks through the paths  $C_{j}^1,\ldots,C_{j}^\tau$ allows us to pebble the path $M_j$ ($j \leq n$). These paths each have length $cm$ and we keep at most one pebble on $M_j$ at any point in time. There may be a delay of up to $c$ time steps between consecutive walks through the paths  $C_{j}^1,\ldots,C_{j}^\tau$ during which we cannot progress our pebble through the path $M_j$. However, there are at most $2cm$ total steps in which we have a pebble on $M_j$. Thus, the cumulative complexity across all paths $M_1,...,M_n$ is at most $2cmn$.

Likewise, since each equation gadget $E_j$ is represented by a chain, one pebble at each time step suffices for each of these paths. We may have to keep a pebble on the gadget $E_j$ while we make $m$ walks through the paths, but we keep this pebble on $E_j$ for at most $2cm$ steps (each walk takes $c$ steps and the delay between consecutive walks is at most $c$ steps).   
Since there are $k$ equations, and there is a chain for each equation, the cumulative complexity across all equation chains is at most $2ckm$.

There is one final sink node, so the cumulative complexity for an instance of \BTwoLC with a valid solution is at most $\tau cmn + 2cmn + 2ckm + 1$.
\end{proofof}

\claimref{claim:sync} will be useful in our proof of \lemmref{no:solution}.
\begin{Claim}
\claimlab{claim:sync}
Any pebbling strategy $P= \left(P_0,\ldots,P_t\right) \in \pPeb_{G_{\BTwoLC}}$ with $\pcc(P) = \pcc(G)$ must satisfy the following property: for all $i \in [n], j \in [\tau]$ and all pebbling rounds $y \in [t]$ and $z \in [c]$ we have $v_i^{j,z} \in P_y \leftrightarrow  \{v_i^{j',z}: ~j' \in [\tau]\} \subseteq P_y$. In particular, whenever we have a pebble on node $v_i^{j,z}$ (the $z$'th node in the $j$'th path gadget $C_i^j$ for variable $x_i$) we also have pebbles on each of the nodes $v_i^{1,z},\ldots,v_i^{\tau,z}$. 
\end{Claim}

\begin{proofof}{\claimref{claim:sync}}
Let $P= \left(P_0,\ldots,P_t\right) \in \pPeb_{G_{\BTwoLC}}$ be a pebbling that does not satisfy our property. We will construct another legal pebbling $P'= \left(P_0',\ldots,P_t'\right)\in \pPeb\left(G_\BTwoLC\right)$ with $\pcc(G) \leq \cc(P') < \cc(P)$. 
For time step $y$, set
\[P_y' = P_y \setminus UNSYNC_y \]
where
\[ UNSYNC_y =\left\{ v_i^{j,z} : i\in[n],j\in[\tau],z\in[c] \text{ and } \{ v_i^{j',z} : j' \in [\tau]\}  \not\subset P_y  \right\} \ . \]
We clearly have $\left|P_y' \right| \leq \left| P_y\right|$ at each time step $y$. Furthermore, because $P$ does not satisfy our property we must have $\left| P_y \right| >  \left|P_y' \right|$ for some step $y$. Thus, \[\cc(P')=\sum|P_y'| < \sum|P_y| = \cc(P) \ .\] It remains to show that $P'$ is a legal pebbling i.e.,  $\forall y \forall v \in P'_{y+1}$ we have
$\parents(P'_{y+1}) \subseteq P_y'$. 
For a node $v \in P'_{y+1}$, we have four cases:
\begin{enumerate}
\item
For some variable $x_i$ the node $v$ is a part of one of the $\tau$ gadgets for that variable i.e.,  $v = v_i^{j,z} \in C_i^j$ for some $j \in [\tau]$. 
By construction of $P'$ we must also have $ v_i^{j',z} \in P_{y+1}'$ for each $j'\in [\tau]$. 
Since $v_i^{j',z} \in P_{y+1}' \subseteq P_{y+1}$, we must have $\parents(v_i^{j',z}) = \{ v_i^{j',z-1} \} \subset P_y$ by the legality of the original pebbling P. 
Thus, $\{ v_i^{j',z-1} : j' \in [\tau]\} \subseteq P_y'$ since $UNSYNC_y \cap  \{ v_i^{j',z-1} : j' \in [\tau]\} = \emptyset$.  
It follows that $\parents(v) \subseteq P_y'$.
\item
$v = z_i^{p+qc} \in M_i$. In this case we observe that, since $v \in P_{y+1}$, we have $\parents(v) \subseteq P_y$ by the legality of the original pebbling $P$. By construction, of $G_{B2LC}$ we have $\parents(z_i^{p+qc}) = \{v_i^{j',p} : j \in [\tau] \} \cup \{z_i^{ p+qc -1}\}$. In the construction of $P_y'$ we would not discard any of these pebbles since $\parents(z_i^{p+qc})  \cap UNSYNC_y = \emptyset$. Thus,  we have $\parents(v) \subseteq P_y'$.
\item
$v = e_i^k \in E_i$ for some equation gadget $E_i$. The proof that $\parents(v) \subseteq P_y'$ is essentially the same as in case 2.
\item
$v$ is a sink. By legality of $P$ we have $\parents(v) \subseteq P_y$. In this case we note that parents(sink) is disjoint from all of the variable gadgets $C_i^j$. Since, $P_y' = P_Y\setminus UNSYNC_y$ can only remove pebbles on the variable gadgets $C_i^j$ it follows that $\parents(v) \subseteq P_y'$.
\end{enumerate}
In each case, $\parents(v) \subseteq P_y'$ so $P'$ is a legal pebbling.
\end{proofof}

\begin{reminderlemma}{\lemmref{no:solution}}
\LemmaNoSolution
\end{reminderlemma}
\newline
\begin{proofof}{\lemmref{no:solution}}
Suppose the given instance of \BTwoLC does not have a valid solution and let $P=(P_0,\ldots,P_t) \in \pPeb_{G_{\BTwoLC}}$  be given such that $\cc(P) = \pcc\left(G_{\BTwoLC} \right)$ i.e. $P$ is optimal. We first observe that in any legal pebbling $P=(P_0,\ldots,P_t) \in \pPeb\left(G_{\BTwoLC} \right)$ we  must walk a pebble down each of the paths $M_1,\ldots,M_n$ of length $cm$. Let $t_i^z$ denote the first time step in which we place a pebble on $v_i^z$ --- the $z$'th node on path $M_i$. Clearly, $t_i^z < t_i^{z+1}$ for $z < cm$ and at time $t_i^z-1$ we must have a pebble on all nodes in $\parents\big(v_i^z\big)$. In particular, $v_i^{1,y},\ldots,v_i^{\tau,y} \in P_{t_{i}^z-1}$ where $y = z \pmod c$. Thus, for each node $v_i^{j,z}$ with $i \in [n], j \in [\tau], z \in [c]$ there are at least $m$ distinct rounds $y \in \left\{t_i^{z},t_i^{z+c},\ldots,t_i^{z+c(m-1)}\right\}$ during which $v_i^{j,z} \in P_y$. 

We say that the pebbling $P$ ``cheats'' if it does not correspond to a valid  \BTwoLC solution. Formally, $P$ is a ``cheating'' pebbling if for some node $v_i^{j,z}$ there are at least $m+1$ distinct rounds $y \in \{ y_1, \ldots , y_{m+1}\}$ during which $v_i^{j,z} \in P_y$. By Claim \claimref{claim:sync} we must have  $\left\{v_i^{1,z},\ldots, v_i^{\tau,z}\right\} \in P_y$ for each $y \in \{ y_1, \ldots , y_{m+1}\}$. Thus,  if $P$ is a cheating pebbling we have 
\begin{eqnarray*} \cc(P)  &\geq& \sum_{y \in [t]}  \sum_{j \in [\tau]} \left| P_y \cap v_i^{j,z}\right| + \sum_{y \in [t]} \sum_{j \in [\tau]} \sum_{\stackrel{i \in [n],z \in [c]~s.t}{(i',z') \neq (i,z)}}  \left| P_y \cap v_i^{j,z}\right|\\
&\geq&\tau (m+1) + \sum_{j \in [\tau]} \sum_{\stackrel{i \in [n],z \in [c]~s.t}{(i',z') \neq (i,z)}} m \\
&=& \tau (m+1)+ \left( cn\tau m - \tau m\right) \\ 
&=& \tau cmn + \tau  \ .  \end{eqnarray*}

Any non-cheating pebbling corresponds to a valid \BTwoLC solution. If an equation $x_{\alpha_i} + c_i + x_{\beta_i}$ is not satisfied by one of the assignments in the \BTwoLC  solution then there is no legal way to pebble the equation chain $E_i$ without cheating because at no point during the $m$ walks are both variables in the equation offset by the correct amount. Thus, any instance of \BTwoLC which does not have a valid solution requires that $\pcc\left(G_{\BTwoLC} \right) \geq \tau cmn + \tau$.
\end{proofof}

\begin{remindertheorem}{\thmref{minCCd}}
\thmminCCd
\end{remindertheorem}
\newline
\begin{proofof}{\thmref{minCCd}} (sketch)
We sketch the construction for $\delta=2$ due to its similarity to the general case where the indegree is not restricted and we highlight the differences between the constructions.  
Although we maintain $\tau$ chains representing each variable, we can no longer maintain the gadgets for the equation, $E_1,\ldots,E_m$, for which each vertex has indegree $2\tau+1$, one edge from its predecessor in the gadget, and $2\tau$ edges from the chains representing the variables involved in the equation. 
Instead, in place of the $2\tau$ edges, we construct a binary tree, where the bottom layer of the tree has at least $\binom{\tau}{2}$ nodes.  
Each of the $\binom{\tau}{2}$ connects to a separate instance of the chains representing the variables involved in the equation. 
Thus, if the equation involves variables $x_i$ and $x_j$, then each of the $\binom{\tau}{2}$ nodes will have an edge from one of the $\tau$ chains representing $x_i$, and an edge from one of the $\tau$ chains representing $x_j$, offset by an appropriate amount. 
This construction ensures that the root of the tree is pebbled only if all equations are satisfied, but any ``dishonest'' walk along the chains will require at least $\tau$ additional pebbles. 
Similarly, we replace the $2\tau$ edges in the paths $P_1,\ldots,P_n$ of length $cm$ with binary trees with base $\binom{\tau}{2}$. 
Finally, we replace the $m+n$ incident edges to the sink node with a binary tree with base $\binom{m+n}{2}$. 
Since each of these terms are polynomial in $c,m,n$, there exists a $\tau$ that is also polynomial in $c,m,n$ which allows us to distinguish between honest walks and dishonest walks. 
As a result, we can decide between instances of $\BTwoLC$.
\end{proofof}

\begin{remindertheorem}{\thmref{reduce:3part}}
\thmreduceThreePart
\end{remindertheorem}
\newline
\begin{proofof}{\thmref{reduce:3part}}
Given an instance $S$ of \textsf{3-PARTITION}, first sort $S$ so that $S=\{s_1,s_2,\ldots,s_m\}$ and $s_i\le s_j$ for any $i\le j$. Let $T=\sum_{i=1}^m s_m$. We then create $mn=3n^2$ equations:
\begin{gather*}
x_1+s_1=x_2,\quad x_2+s_2=x_3,\quad\ldots,\quad x_m+s_m=x_{m+1},\\
x_1+0=x_2,\quad x_2+0=x_3,\quad\ldots,\quad x_m+0=x_{m+1},\\
x_1+T=x_2,\quad x_2+T=x_3,\quad\ldots,\quad x_m+T=x_{m+1},\\
x_1+2T=x_2,\quad x_2+2T=x_3,\quad\ldots,\quad x_m+2T=x_{m+1},\\
x_1+(n-2)T=x_2,\quad x_2+(n-2)T=x_3,\quad\ldots,\quad  x_m+(n-2)T=x_{m+1},
\end{gather*}
Finally, we create the additional $n$ equations:
\begin{equation}
\label{eq:constraints}
x_1+\frac{T}{n}+3(i-1)(n-2)T=x_{m+1},
\end{equation}
for $i\in[n]$. This gives a total of $3n^2+n$ equations so that the reduction is clearly achieved in polynomial time.

Recall that \textsf{B2LC} is true if and only if there exist $\{a_{i,j}\}$, $i\in[n]$, $j\in[m]$ so that each equation $x_i+c_{i,j}=x_j$ is satisfied by  the assigning $x_i=a_{i,k}$ and $x_j=a_{j,k}$ for some $k$. We show that there exists a solution to \textsf{3-PARTITION} if and only if there exists a solution to \textsf{B2LC}.

Suppose there exists a partition of $S$ into $n$ triplets $S_1,S_2,\ldots,S_n$ so that the sum of the integers in each triplet is the same, and equals $\frac{T}{n}$.
We set $a_{1,1}=0$ and for each $s_i$ that appears in $S_1$, we choose to satisfy the equation, $a_{1,i}+s_i=a_{1,i+1}$. 
Otherwise, if $s_i$ does not appear in $S_1$, we choose to satisfy the equation $a_{1,i}+0=a_{1,i+1}$.

Suppose that $a_{i,j}$ are defined for all $i<k$.
Then we set $a_{k,1}=0$ and for each $s_i$ that appears in $S_k$, we choose to satisfy the equation, $a_{k,i}+s_i=a_{k,i+1}$.
If $s_i$ appears in $S_i$ for some $i<k$, then we choose to satisfy the equation $a_{k,i}+(i-2)T=a_{k,i+1}$.
Otherwise, if $s_i$ does not appear in $S_1,\ldots,S_k$, we choose to satisfy the equation $a_{k,i}+(i-1)T=a_{k,i+1}$.

Thus, to get $a_{i,m+1}$ from $a_{i,1}$, we add in three elements whose sum is $\frac{T}{n}$. 
That is, we pick three unused elements of $S$, say $s_t,s_u,s_v$, and we choose to satisfy the equations
$x_t+s_t=x_{t+1}$, $x_u+s_u=x_{u+1}$, and $x_v+s_v=x_{v+1}$. 
We then add in $3(i-1)$ instances of $(i-2)T$, for each of the elements which appear in $S_1,\ldots,S_{i-1}$. 
Finally, we add in $(3n-3i)$ instances of $(i-1)T$.
Hence, it follows that 
\begin{align*}
a_{i,1}+\frac{T}{n}+3(i-1)(i-2)T+(3n-3i)(i-1)T&=a_{i,m+1}\\
a_{i,1}+\frac{T}{n}+(3in-3n-6i+6)T&=a_{i,m+1}\\
a_{i,1}+\frac{T}{n}+3(i-1)(n-2)T&=a_{i,m+1},\\
\end{align*}
so that all equations in \ref{eq:constraints} are satisfied. Thus, a solution for \textsf{3-PARTITION} yields a solution for \textsf{B2LC}.

Suppose there exists a solution for the above instance of \textsf{B2LC}. 
Observe that since there are $n$ instances of the equation $x_1+\frac{T}{n}+3(i-1)(n-2)T=x_{m+1}$, then the equation must hold for each $a_{i,1},a_{i,m+1}$. 
Thus, $a_{i,1}+\frac{T}{n}\equiv a_{i,m+1}\pmod{T}$ for all $i$.
Hence, to obtain $a_{1,m+1}$ from $a_{1,1}$, we must take three equations of the form $x_i+s_i=x_{i+1}$ since the summing less than three elements in $S$ is less than $\frac{T}{n}$, while the summing more than three elements in $S$ is more than $\frac{T}{n}$ but less than $T+\frac{T}{n}$ (as each element of $S$ is greater than $\frac{T}{4n}$ and less than $\frac{T}{2n}$ and the sum of all elements in $S$ is $T$). 
Say that these three equations use the terms $s_{i_1}, s_{i_2}, s_{i_3}$. 
Then we let $S_1=\{s_{i_1}, s_{i_2}, s_{i_3}\}$, which indeed sums to $\frac{T}{n}$.

Similarly, to obtain $a_{k,m+1}$ from $a_{k,1}$, we must take three equations of the form $x_i+s_i=x_{i+1}$. 
Since there are $3n$ equations of this form which must be satisfied and each of the $n$ assignments of the form $a_{i,1},\ldots,a_{i,m+1}$ (where $1\le i\le n$) uses exactly three of these equations, then each assignment of $a_{i,1},\ldots,a_{i,m+1}$ must satisfy a disjoint triplet of the $3n$ equations. 
Say that the three satisfied equations for $a_{i,1},\ldots,a_{i,m+1}$ are $s_{k_1}, s_{k_2}, s_{k_3}$. 
Then we let $S_k=\{s_{k_1}, s_{k_2}, s_{k_3}\}$, which indeed sums to $\frac{T}{n}$.

Therefore, we have a partition of $S$ into triplets, which each sum to $\frac{T}{n}$, as desired. Thus, a solution for \textsf{B2LC} yields a solution for \textsf{3-PARTITION}.
\end{proofof}

\begin{remindertheorem}{\thmref{reducible:npc}}
\reducibleNPC
\end{remindertheorem}
\newline
\begin{proofof}{\thmref{reducible:npc}}
We provide a reduction from $\VC$ to $\REDUCIBLE_d$. Given an instance $G=(V,E)$ of $\VC$, arbitrarily label the vertices $1,\ldots,n$, where $n=|V|$, and direct each edge of $E$ so that $1,\ldots,n$ is a topological ordering of the nodes. For each node $i$ we add two directed paths (where path length is the number of edges in the path): (1) a path of length $i-1$ with an edge from the last node on the path to node $i$, (2) a path of length $n-i$ with an edge from node $i$ to the first node of the path. To avoid abuse of notation, let $U$ represent the original $n$ vertices (from the given instance of $\VC$) in the modified graph.

We claim $\VC$ has a vertex cover of size at most $k$ if and only if the resulting graph is $(k,n)$-reducible. Indeed, if there exists a vertex cover of size $k$, we remove the corresponding $k$ vertices in the resulting construction. Then there are no edges connecting vertices of $U$. By construction, any path of length $n-1$ must contain at least two vertices of $U$. Hence, the resulting graph is $(k,n)$-reducible. 

On the other hand, suppose that there is no vertex cover of size $k$. Given a set $S$ of $k= |S|$ vertices we say that a node $u \in U$ is ``untouched'' by $S$ if $u \notin S$ and $S$ does not contain any vertex from the chain(s) we connected to node $u$. If there is no vertex cover of size $k$, then removing any set $S$ of  $k=|S|$ vertices from the graph leaves an edge $(u,v)$ with the following properties: (1) $u,v \in U$, (2) $u$ and $v$ are both untouched by $S$. Suppose $u$ has label $i$. Then the (untouched) directed path which ends at $u$ has length $i-1$. Similarly, there still exists some directed path of length $\geq n-i$ which begins at $v$. Thus, there exists a path of length at least $n$, so the resulting graph is not $(k,n)$-reducible.
\end{proofof}

\begin{remindertheorem}{\thmref{IPgap}}
\thmIPgap
\end{remindertheorem}
\newline
\begin{proofof}{\thmref{IPgap}}
In particular, for all time steps $t\le n$, we set $x_i^t=\frac{1}{n}$ for all $i\le t$, and $x_i^t=0$ for $i>t$. 
For time steps $n<t\le n+\ceil{\log n}$, we set $x_v^t=\min\left(1,\frac{2^{t-n}}{n}\right)$ for all $v\in V$. 
We first argue that this is a feasible solution. 
Trivially, Constraints 1 (with the LP relaxation) and 2 are satisfied. 
Moreover, for $t=n+\ceil{\log n}$, $x_v^t=1$ for all $v\in V$, so Constraint 3 is satisfied.
Note that for $1<t\le n$, if $x_v^t=\frac{1}{n}$, then $x_u^{t-1}=\frac{1}{n}$ for all $u<v$, so certainly $\sum_{v' \in\parents(v)} \frac{x_{v'}^{t-1}}{\left|\parents(v) \right|}\ge\frac{1}{n}$. 
Furthermore, $x_v^{t-1}=\min\left(1,\frac{2^{t-1-n}}{n}\right)$ for $n<t\le n+\ceil{\log n}$ implies that setting
\[x_v^t=x_v^{t-1}+\sum_{v' \in\parents(v)} \frac{x_{v'}^t}{\left|\parents(v) \right|}=\frac{2^{t-1-n}}{n}+\frac{2^{t-1-n}}{n}=\frac{2^{t-n}}{n}\]
is valid. 
Therefore, Constraint 4 is satisfied.

Finally, we claim that $\sum_{v \in V} \sum_{t \leq n+\ceil{\log n}} x_{v}^t \leq 4n$. To see this, note that for every round $t\le n$, we have $x_{v}^t \leq \frac{1}{n}$ for all $v\in V$. Thus, \[\sum_{v \in V} \sum_{t \leq n} x_{v}^t \leq \sum_{v \in V} \sum_{t \leq n}\frac{1}{n}\le\sum_{v\in V}1=n.\]
For time steps $n<t<n+\ceil{\log n}$, note that $x_v^t\le\frac{2^{t-n}}{n}$ for all $v\in V$. Thus,
\[\sum_{v \in V} \sum_{n<t<n+\ceil{\log n}} x_{v}^t \leq \sum_{n<t<n+\ceil{\log n}}2^{t-n}\le 2n.\]
Finally, for time step $t=n+\ceil{\log n}$, $x_v^t=1$ for all $v\in V$. Consequently,
\[\sum_{v \in V} \sum_{t=n+\ceil{\log n}} x_{v}^t=n.\] 
Therefore,
\[\sum_{v \in V} \sum_{t \leq n+\ceil{\log n}} x_{v}^t \leq 4n.\]
\end{proofof}

%% file: counterexample.tex

In this section we address the following question: is is necessarily case that an optimal pebbling of $G$ (minimizing cumulative cost) takes $n$ steps? For example, the pebbling attacks of Alwen and Blocki~\cite{AlwenB16a,AlwenB16b} on depth-reducible graphs such as Argon2i-A, Argon2i-B, Catena all take {\em exactly} $n$ steps. Thus, it seems natural to conjecture that the optimal pebbling of $G$ {\em always} finishes in $\depth(G)$ steps. In this section we provide a concrete example of an DAG $G$ (with $n$ nodes and $\depth(G)=n$) with the property that {\em any} optimal pebbling of $G$ {\em must} take more than $n$ steps. More formally, let $\ppeb(G,t)$ denote the set of all legal pebblings of $G$ that take at most $t$ pebbling rounds. We prove that $\min_{P \in \ppeb(G,t)} \cc(P) > \min_{P \in \ppeb(G)} \cc(P) = \pcc(G)$.

\begin{theorem} \thmlab{timecounterexample}
There exists a graph $G$ with $n=16$ nodes and $\depth(G) = n$ s.t. for some $t>0$, $\frac{\min_{P \in \ppeb(G,t)} \cc(P)}{\pcc(G)} \geq \frac{28}{27}$. 
\end{theorem}
\begin{proof}
Consider the following DAG $G$ on $16$ nodes $\{1,\ldots,16\}$ with the following directed edges (1) $(i,i+1)$ for $1 \leq i < 16$, (2) $(i,i+9)$ for $1 \leq i \leq 5$ and (3) $(i,i+7)$ for $8 \leq i \leq 9$. We first show in \claimref{timelowerbound} that there is a pebbling $P$ with cost $27$ that takes $18$ rounds. \thmref{timecounterexample} then follows from \claimref{timeupperbound} where we show that any $P \in \ppeb(G,16)$ has cumulative cost at least $28$. Intuitively, any legal pebbling must either delay before pebbling nodes $15$ and $16$ or during rounds $15-i$ (for $0 \leq i \leq 5$) we {\em must} have at least one pebble on some nodes in the set $\{9,8,...,9-i\}$.

\begin{claim} \claimlab{timelowerbound}
For the above DAG $G$ we have $\pcc(G) \leq 27$.
\end{claim} 
\begin{proof}
Consider the following pebbling: $P_0=\emptyset, P_1=\{1\}, P_2=\{2\}, P_3=\{3\}, P_4=\{4\}, P_5=\{5\}, P_6=\{6\}, P_7=\{7\}, P_\{8\}=\{8\}, P_9 = \{1,9\},  P_{10} = \{2,10\}, P_{11} = \{3,11\}, P_{12} = \{4,12\}, P_{13} = \{5,13\}, P_{14} = \{6,14\}, P_{15} = \{7,14\}, P_{16} = \{8,14\}, P_{17} = \{9,15\}, P_{18} = \{16\}$. It is easy to verify that $\cc(P) = 9+2*9=27$ since there are $9$ steps in which we have one pebble on $G$ and $9$ steps in which we have two pebbles on $G$.
\end{proof}

\begin{claim} \claimlab{timeupperbound}
For the above DAG $G$ we have $\min_{P \in \ppeb(G,16)} \cc(P) \geq 28$.  
\end{claim} 
\begin{proof}
Let $P = (P_1,\ldots,P_{16}) \in \ppeb(G,16)$ be given. Clearly, $i \in P_i$ for each round $1 \leq i \leq 16$.  To pebble nodes $15$ and $16$ on steps $15$ and $16$ we must have $9 \in P_{15}$ and $8 \in P_{14}$. By induction, this means that $P_{14-i} \cap \{8,...,8-i\} \neq \emptyset$ for each $i > 0$. To pebble node $9+i$ at time $9+i$ we require that $ i \in P_{8+i} $ for each $1 \leq i \leq 5$. These observations imply that $\left|P_9 \right| \geq 3$, $\left|P_{10} \right| \geq 3$, $\ldots$, $\left|P_{13} \right| \geq 3$. We also have $\left| P_{15} \right| \geq 2$ and $\left|P_{14}\right| \geq 2$. We also have $\left| P_i \right| \geq 1$ for all $1 \leq i \leq 16$.  The cost of rounds 1--8 and round $16$ is at least $9$. The cost of rounds 9--13 is at least $15$ and the cost of rounds $14$ and $15$ is at least $4$. Thus, $\cc(P) \geq 28$.  
\end{proof}
\end{proof}
{\noindent \bf Open Questions: }\thmref{timecounterexample} shows that $\pcc(G)$ can be smaller than $\min_{P \in \ppeb(G,t)} \cc(P)$, but how large can this gap be in general? Can we prove upper/lower bounds on the ratio: $\frac{\min_{P \in \ppeb(G,t)} \cc(P)}{ \pcc(G)}$ for any $n$ node DAG $G$? Is it true that $\frac{\min_{P \in \ppeb(G,t)} \cc(P)}{ \pcc(G)} \leq c$ for some constant $c$? If not does this hold for $n$ node DAGs $G$ with constant indegree? Is it true that the optimal pebbling of $G$ always takes at most $cn$ steps for some constant $c$?